\documentclass[conference]{IEEEtran}
\usepackage{amsmath,amsfonts,amssymb,amsthm}
\usepackage{multicol,array}
\usepackage{graphicx}
\usepackage{caption}
\usepackage{subcaption}
\usepackage{float}
\usepackage{etoolbox}
\newcommand*{\QEDA}{\hfill\ensuremath{\blacksquare}}

\renewcommand{\thesection}{\arabic{section}}
\theoremstyle{definition}
\newtheorem{defn}{\textit{Definition}}
\newtheorem{cons}{\textit{Construction}}
\newtheorem{claim}{\textit{Claim}}
\newtheorem{algo}{\textit{Algorithm}}
\newtheorem{thm}{\textit{Theorem}}

\newtheorem{ex}{\textit{Example}}[section]
\AtEndEnvironment{ex}{\null\hfill\QEDA}
\newtheorem{exmp}{\textit{Example}}[section]
\AtEndEnvironment{exmp}{\null\hfill\QEDA}
\newtheorem*{example}{\textit{Example}}
\newtheorem{lm}{\textit{Lemma}}

\theoremstyle{remark}
\newtheorem{remark}{Remark}
\newtheorem{cond}{\textit{Condition}}

\title{Index Codes for Interlinked Cycle Structures with Outer Cycles}
\author{K. Vikas Bharadwaj and B. Sundar Rajan, {\it Fellow, IEEE}}
\begin{document}	\maketitle
\begin{abstract}Index code construction for a class of side-information graphs called interlinked cycle (IC) structures without outer cycles is given by Thapa, Ong and Johnson in \cite{TOJ,TOJ2} (C. Thapa, L. Ong, and S. Johnson, ``Interlinked Cycles for Index Coding: Generalizing Cycles and Cliques", in \textit{IEEE Trans. Inf. Theory, vol. 63, no. 6, Jun. 2017}, ``Interlinked Cycles for Index Coding: Generalizing Cycles and Cliques", in arxiv (arxiv:1603.00092v2 [cs.IT] 25 Feb 2018)). In this paper, construction of index codes for interlinked cycle (IC) structures with outer cycles is given along with a decoding algorithm.
		\footnote{The authors are with the Department of Electrical Communication Engineering, Indian Institute of Science, Bangalore-560012, India. Email:bsrajan@iisc.ac.in}
	\end{abstract}
	\section{Introduction}
	\label{sec_int}
	The problem of index coding was introduced by Birk and Kol in \cite{bk}. The index coding problem consists of a single sender with a set of $ N $ independent messages $\mathcal{X}=\lbrace x_1,x_2,\dots,x_N\rbrace$, and a set of $ M $ users $ \mathcal{D}=\lbrace D_1,D_2,\dots,D_M \rbrace $, connected to the sender by a single shared error-free link, with the $ k^{th} $ user identified as $	D_k=(\mathcal{X}_k,\mathcal{A}_k) $, where $ \mathcal{X}_k \subseteq \mathcal{X}$ is the set of messages desired by $ D_k $, the set $ \mathcal{A}_k \subset \mathcal{X}$ is comprised of the messages available to user $ D_k $ as side information. The set $ \mathcal{A}_k $ satisfies $ \mathcal{X}_k \cap \mathcal{A}_k=\phi$. In a scalar linear index coding scheme, each message $ x_i \in \mathcal{F}_q $, $ i\in [M] $ and $ \mathcal{F}_q $ is a finite field. Sender encodes $ M $ messages to $ n $ symbols in $ \mathcal{F}_q $ using a linear mapping and $ n $ is called the length of the index code. Each user will decode their desired messages using linear combinations of transmitted $ n $ symbols and side information available to them. An index coding problem is said to be unicast \cite{Ong} if $ \mathcal{X}_k \cap \mathcal{X}_j = \phi $ for $ k\not= j $ and $ k,j\in \lbrace1,2,\dots,M\rbrace $, i.e., no message is desired by more than one user. The problem is said to be single unicast if the problem is unicast and $ |\mathcal{X}_k|=1$ for all $ k\in \lbrace1,2,\dots,M\rbrace $. A unicast index coding problem can be reduced into single unicast index coding problem, by splitting the user demanding more than one message into several users, each demanding one message and with the same side information as the original user. Single unicast index coding problems can be described by a directed graph called as  side information graph \cite{BK2}, in which the vertices in the graph represent the indices of messages $ \lbrace x_1,x_2,\dots,x_N\rbrace $ and there is a directed edge from vertex $ i $ to vertex $ j $ if and only if the user requesting $ x_i $ has $ x_j $ as side information.
The set of vertices in a directed graph $ \mathcal{G} $ is denoted by $ V(\mathcal{G}) $ and the set of vertices in the out-neighborhood of a vertex $q$ in $\mathcal{G}$ is denoted by $N^{+}_{\mathcal{G}}(q)$.

	Interlinked Cycle Cover (ICC) scheme is proposed as a scalar linear index coding scheme to solve single unicast index coding problems by Thapa et al. \cite{TOJ}, by defining a graph structure called an Interlinked Cycle (IC) structure. 
After a correction to the definition of IC structure by Thapa et al. in \cite{TOJ2} the definition of an IC structure is as follows.
	\begin{defn}[\textbf{IC Structure} \cite{TOJ2}]
		\label{def_ICS}
		A side information graph $\mathcal{G}$ is called a $K$-IC structure with inner vertex set $V_{I}$ $\subseteq V(\mathcal{G})$, such that $|V_{I}| = K$ if $ \mathcal{G} $ satisfies the following four conditions.
		\begin{enumerate}
			\item There is no I-Cycle in $\mathcal{G}$, where I-Cycle is defined as a cycle which contains only one inner vertex.
			\item There is a unique I-Path between any two different inner vertices in $\mathcal{G}$, where an I-path is defined as a path from one inner vertex to another inner vertex without passing through any other inner vertex (as a result, $K$ rooted trees can be drawn where each rooted tree is rooted at an inner vertex and has the remaining inner vertices as the leaves).
			\item $\mathcal{G}$ is the union of the $K$ rooted trees.
			\item There are no cycles in $ \mathcal{G} $ containing only non-inner vertices (called as outer cycles).
		\end{enumerate}
		The set of the vertices $ V(\mathcal{G})\backslash V_I $ is called the set of non-inner vertices, denoted by $ V_{NI} $. Let $ V_{NI}(i) $ be the set of non-inner vertices that are present in the rooted tree $ T_i $ of the inner vertex $ i $.
	\end{defn}
	An index code construction (presented as \textit{Construction} $ 1 $ in this paper) and a decoding algorithm (presented as \textit{Algorithm} $ 1 $ in this paper) are also given in \cite{TOJ}.
	Let the $K$-IC structure be called $\mathcal{G}$ and let $|V(\mathcal{G})|=N$. Let $ V(\mathcal{G})=\lbrace1,2,\dots,N\rbrace$, $V_I=\lbrace1,2,\dots,K)$ be the set of the $K$ inner vertices and hence $ V_{NI}=\lbrace K+1, K+2,\dots,N\rbrace $. Let $x_n \in \mathbb{F}_q$ be the message corresponding to the vertex $n \in V(\mathcal{G})$ and where $ \mathbb{F}_q $ is the finite field of characteristic $ 2 $ to which the all the $ N $ messages at the sender belong to (note that in single unicast setting, the number of messages will be equal to the number of users, i.e., $ N=M $ ).
	\begin{cons}[\cite{TOJ}]
		\label{cons1} Given the inner and non-inner vertices, the following coded symbols are transmitted.
		\begin{enumerate}
			\item An index code symbol $W_I$ obtained by XOR of messages corresponding to inner vertices is transmitted, where
			\begin{equation*}
			W_I=\underset{i=1}{\overset{K}{\bigoplus}} x_i.
			\end{equation*}
			\item An index code symbol corresponding to each non-inner vertex, obtained by XOR of message corresponding to the non-inner vertex with the messages corresponding to the vertices in the out-neighborhood of the non-inner vertex is transmitted, i.e., for $j \in V_{NI}$, $W_j$ is transmitted, where
			\begin{equation*}
			W_j=x_j \underset{q \in N^{+}_{\mathcal{G}}(j)}{\bigoplus} x_q,
			\end{equation*}
		\end{enumerate}where $ \oplus $ denotes modulo addition over $ \mathbb{F}_q $.
	\end{cons}
	\begin{algo}
		\label{algo1} It is the algorithm proposed in \cite{TOJ} to decode an index code obtained by using \textit{Construction} $ 1 $ on an IC structure, $ \mathcal{G} $.
		\begin{itemize}
			\item The message $x_j$ corresponding to a non-inner vertex $j$ is decoded directly using the transmission $W_j$ and
			\item the message $x_i$ corresponding to an inner vertex $i$ is decoded using
			\begin{displaymath}
			Z_i=W_I \underset{q \in V_{NI}(i)}{\bigoplus}W_q
			\implies
			Z_i=x_i \underset{k \in N^{+}_{T_{i}}(i)}{\bigoplus}x_k.
			\end{displaymath}
		\end{itemize}
	\end{algo}
	IC structures which contain outer cycles, i.e., side-information graphs which satisfy only the first three conditions of \textit{Definition} \ref{def_ICS} are considered in this paper. Since these are the side-information graphs considered in \cite{vikas}, the results obtained in \cite{vikas} are applicable for IC structures with outer cycles. A short summary of the results obtained in \cite{vikas} is as follows.
	\begin{itemize}
		\item Given an IC structure $ \mathcal{G} $ containing outer cycles and inner vertex set $ V_I=\lbrace1,2,\dots,K \rbrace $, for each $ i \in \lbrace1,2,\dots,K\rbrace $ and for a non-inner vertex $ j $ which is at a depth $ \geq2 $ in the rooted tree $ T_i $, define $ a_{i,j} $ as the number of vertices in $V_{NI}(i)$ for which $ j $ is in out-neighborhood in $ \mathcal{G} $, i.e., for each $ i \in \lbrace1,2,\dots,K\rbrace $ and for $j \in {V_{NI}(i) \backslash N^+_{T_i}(i)} $, $$ a_{i,j} \triangleq |\lbrace v:v\in V_{NI}(i), j\in N^{+}_{\mathcal{G}}(v)\rbrace| .$$ Also, for each $ i \in \lbrace1,2,\dots,K\rbrace $ and for a non-inner vertex $ j $ not in the rooted tree $ T_i $, define $ b_{i,j} $ as the number of vertices in $ V_{NI}(i) $ for which $ j $ is in out-neighborhood in $ \mathcal{G}$, i.e., for each $ i \in \lbrace1,2,\dots,K\rbrace $ and $ j \in V(\mathcal{G}) \backslash V(T_i) $, $$ b_{i,j} \triangleq |\lbrace v:v\in V_{NI}(i),~j\in N^{+}_{\mathcal{G}}(v)\rbrace| .$$ 
		
		\item  It is shown that $ b_{i,j} \in \lbrace 0,1\rbrace $ for each $ i \in \lbrace1,2,\dots,K\rbrace $ and $ j \in V(\mathcal{G}) \backslash V(T_i) $.
		\item The index code obtained by using \textit{Construction} $1$ on an IC structure $ \mathcal{G} $ is decodable using \textit{Algorithm} $1$ if and only if the IC structure, $ \mathcal{G} $, satisfies the following two conditions \textit{c}$1$ and \textit{c}$2$.
		\begin{cond}[\textit{c}$1$] $ a_{i,j} $ must be an odd number for each $ i \in \lbrace1,2,\dots,K\rbrace $ and $ j \in V_{NI}(i)\backslash N^+_{T_i}(i)$.  
		\end{cond}
		\begin{cond}[\textit{c}$2$] $ b_{i,j} $ must be zero for each $ i \in \lbrace1,2,\dots,K\rbrace $ and $ j \in V(\mathcal{G}) \backslash  V(T_i)$.
		\end{cond}
		\item Since \textit{c}$ 1 $ and \textit{c}$ 2 $ are true for an IC structure without outer cycles, the index code obtained by using \textit{Construction} $ 1 $ on an IC structure is decodable using \textit{Algorithm} $ 1 $.
\end{itemize}

The rest of the paper is organized as follows. The following section discusses the main results along with illustrating examples. Section \ref{sec_ex} illustrates the results using another example and Section \ref{sec_con} provides the conclusion.
\section{Main Results}
\label{sec_main}
	Let the IC structure with outer cycles be called $ \mathcal{G} $. Let the set of vertices present in an outer cycle is called as a vertex set of that outer cycle. 
\subsection{Definitions}
\begin{defn}[An outer-cycle group]
		Consider a vertex $ j $ which is present in at least two outer cycles. The subgraph which is formed by the union of all outer cycles which contain $ j $ is called the outer-cycle group of $ j $.
	\end{defn}
	\begin{example}
		Consider Fig. \ref{example1}. Let the cycles of this figure be all the  outer cycles in an IC structure $ \mathcal{G} $  (there are no other outer cycles in $ \mathcal{G} $). Vertices $ 2 $ and $ 3 $ are present in at least two outer cycles. The outer-cycle group of vertex $ 2 $ is as shown in Fig. \ref{example1}. The outer-cycle group of vertex $ 3 $ is shown in Fig. \ref{example2}.
		\begin{figure*}[!t]
			\centering
			\begin{subfigure}{0.49\textwidth}
				\centering
				\includegraphics[width=3in,height=2in,angle=0]{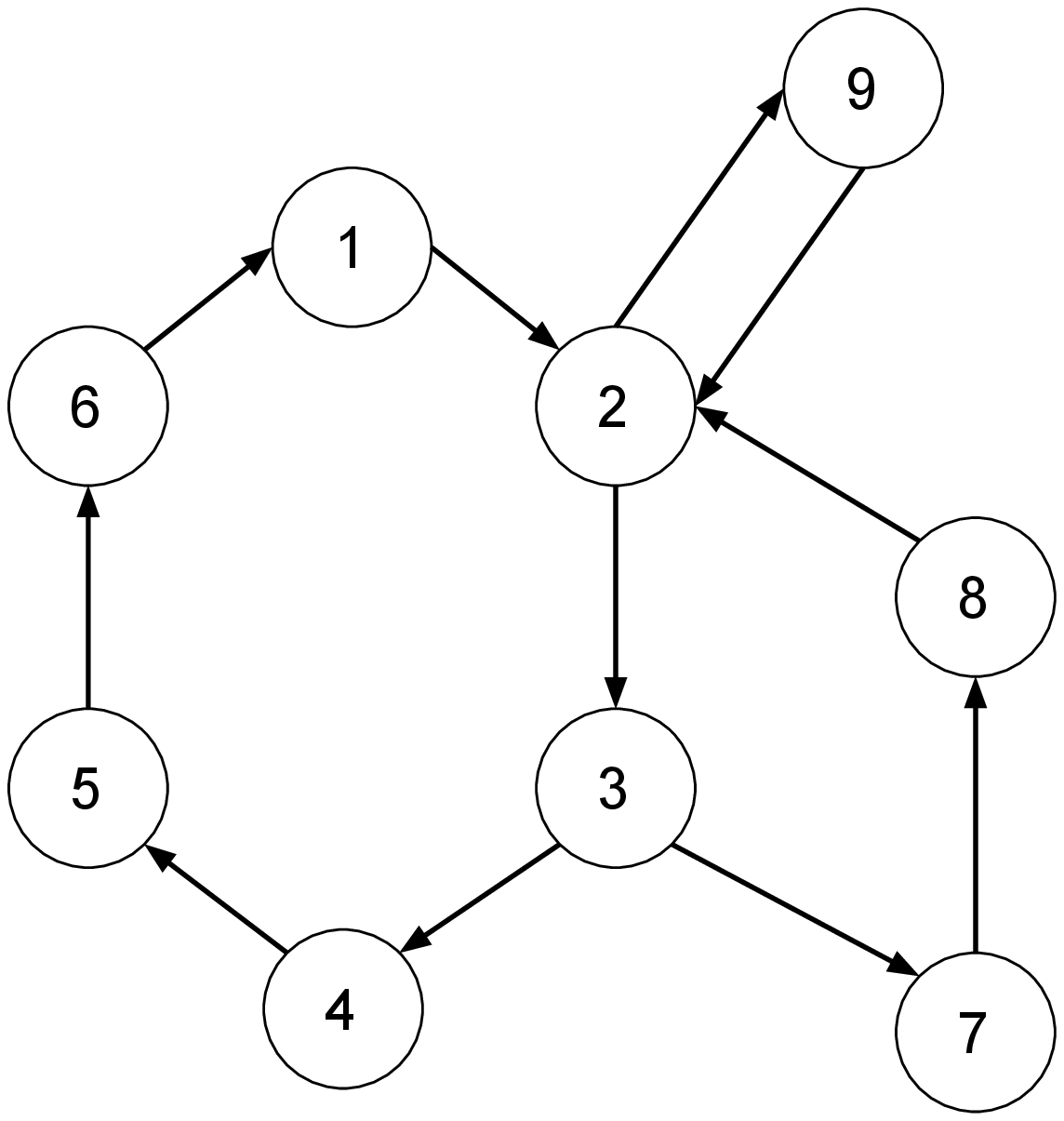}
				\caption{}
				\label{example1}
			\end{subfigure}
			\begin{subfigure}{0.49\textwidth}
				\centering
				\includegraphics[width=3in,height=2in,angle=0]{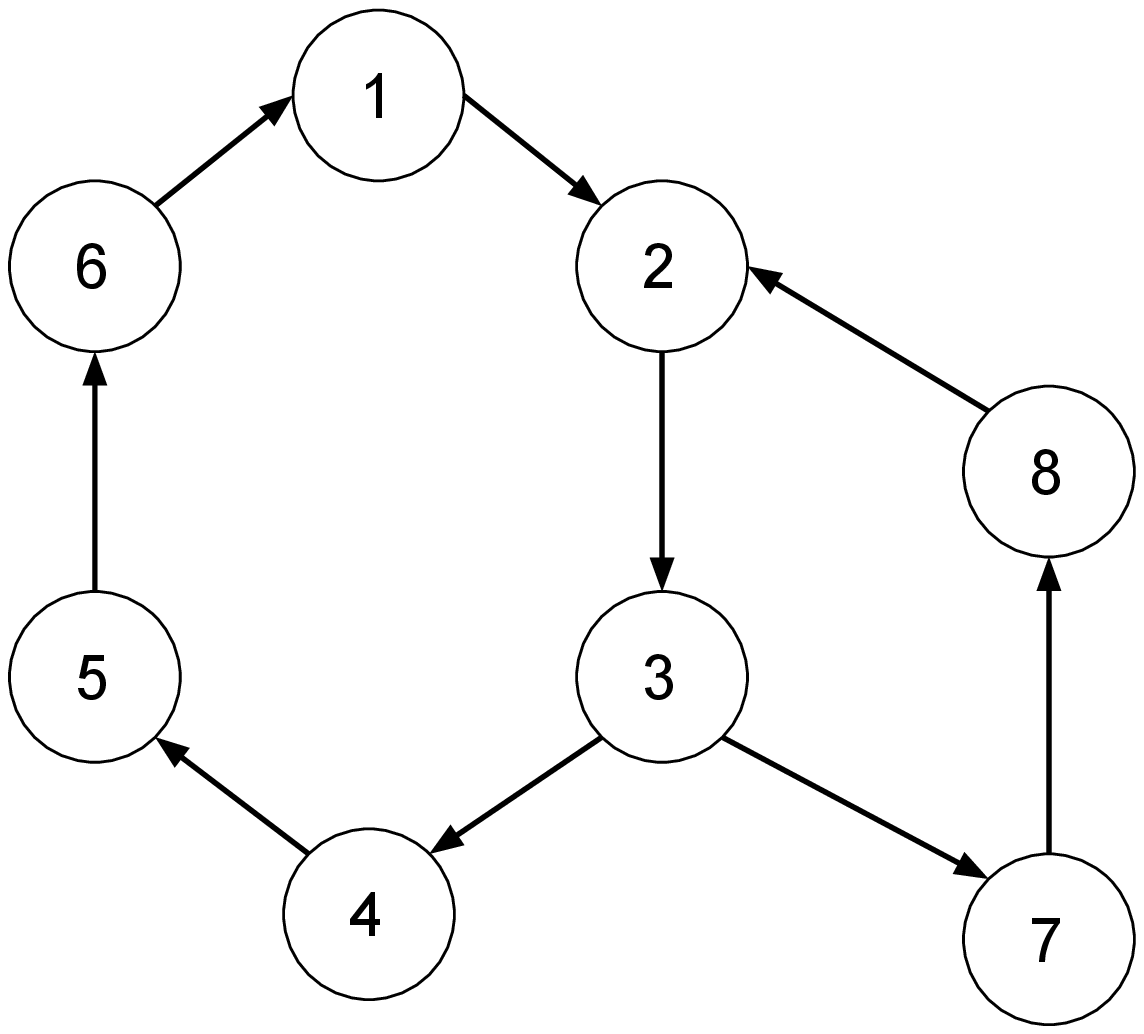}
				\caption{}
				\label{example2}
			\end{subfigure}
			\caption{Figures showing outer-cycle groups of vertices $ 2 $ and $ 3 $ respectively.}
		\end{figure*}
	\end{example}
	\begin{defn}[A maximal outer-cycle group (MOCG) and its central cycle vertex (CCV)]
		Consider a non-inner vertex $ j $ which is present in at least two outer cycles. If the outer-cycle group of $ j $ is not a subgraph of outer-cycle group of any other non-inner vertex, then the outer-cycle group of $ j $ is called a maximal outer-cycle group (MOCG) and the vertex $ j $ is called the central cycle vertex (CCV) corresponding to that MOCG. In the case where more than one non-inner vertices have same outer-cycle group and if that outer-cycle group is not subgraph of any other outer-cycle group, then the outer-cycle group is also called an MOCG and its CCV is the vertex with highest in-degree in the MOCG.
	\end{defn}
	\begin{example}
		From Fig. \ref{example1} and \ref{example2}, it can be seen that the outer-cycle group of vertex $ 3 $ is not an MOCG because it is subgraph of outer-cycle group of vertex $ 2 $. Since the outer-cycle group of vertex $ 2 $ is not a subgraph of any other outer-cycle group, Fig. \ref{example1} is an MOCG with CCV $ 2 $.
	\end{example}
	\begin{defn}[Pre-central cycle vertex set of a CCV]
		Consider an MOCG. Let the MOCG be union of $ n\geq2 $ outer cycles. Let the CCV of the MOCG be $ j $. The set of $ n $ vertices of the MOCG which have $ j $ in their out-neighborhood is called the pre-central cycle vertex set of $ j $ and the $ n $ vertices are called pre-central cycle vertices of $ j $.
	\end{defn}
	\begin{example}
		Consider the MOCG shown in Fig. \ref{example1}. The CCV of that MOCG is $ 2 $. It can be seen that the MOCG is union of $ 3 $ outer cycles and the pre-central cycle vertex set is $ \lbrace1,8,9\rbrace $.
	\end{example}
	\begin{defn}[Isolated and Non-Isolated MOCGs]
		If the outer cycles present in an MOCG are not present in any other MOCG, then the MOCG is called an isolated MOCG. Else, the considered MOCG is called a non-isolated MOCG.
	\end{defn}
	\begin{example}
		Consider Fig. \ref{non_iso}. Let the cycles be outer cycles in an IC structure $ \mathcal{G}' $. (It is assumed that there are no other outer cycles in $ \mathcal{G}' $). Vertices $ 2 $, $ 5 $ and $ 6 $ are present in at least two outer cycles. The outer-cycle group of vertex $ 2 $ is shown in Fig. \ref{example3} and the outer-cycle groups of vertices $ 5 $ and $ 6 $ are both the same and are as shown in Fig. \ref{example4}. Since, one outer-cycle group is not subgraph of the other outer-cycle group, there are two MOCGs, one with CCV $ 2 $ and one with CCV $ 6 $, in Fig. \ref{non_iso}. The MOCG with CCV $ 2 $ is as shown in Fig. \ref{example3} and the MOCG with CCV $ 6 $ is as shown in Fig. \ref{example4} (CCV is $ 6 $ because $ 6 $ has an in-degree of $ 2 $ and $ 5 $ has an in-degree of $ 1 $ in the MOCG).
		\begin{figure*}[!t]
			\centering
			\begin{subfigure}{0.31\textwidth}
				\centering
				\hspace*{-12mm}
				\includegraphics[width=3in,height=2in,angle=0]{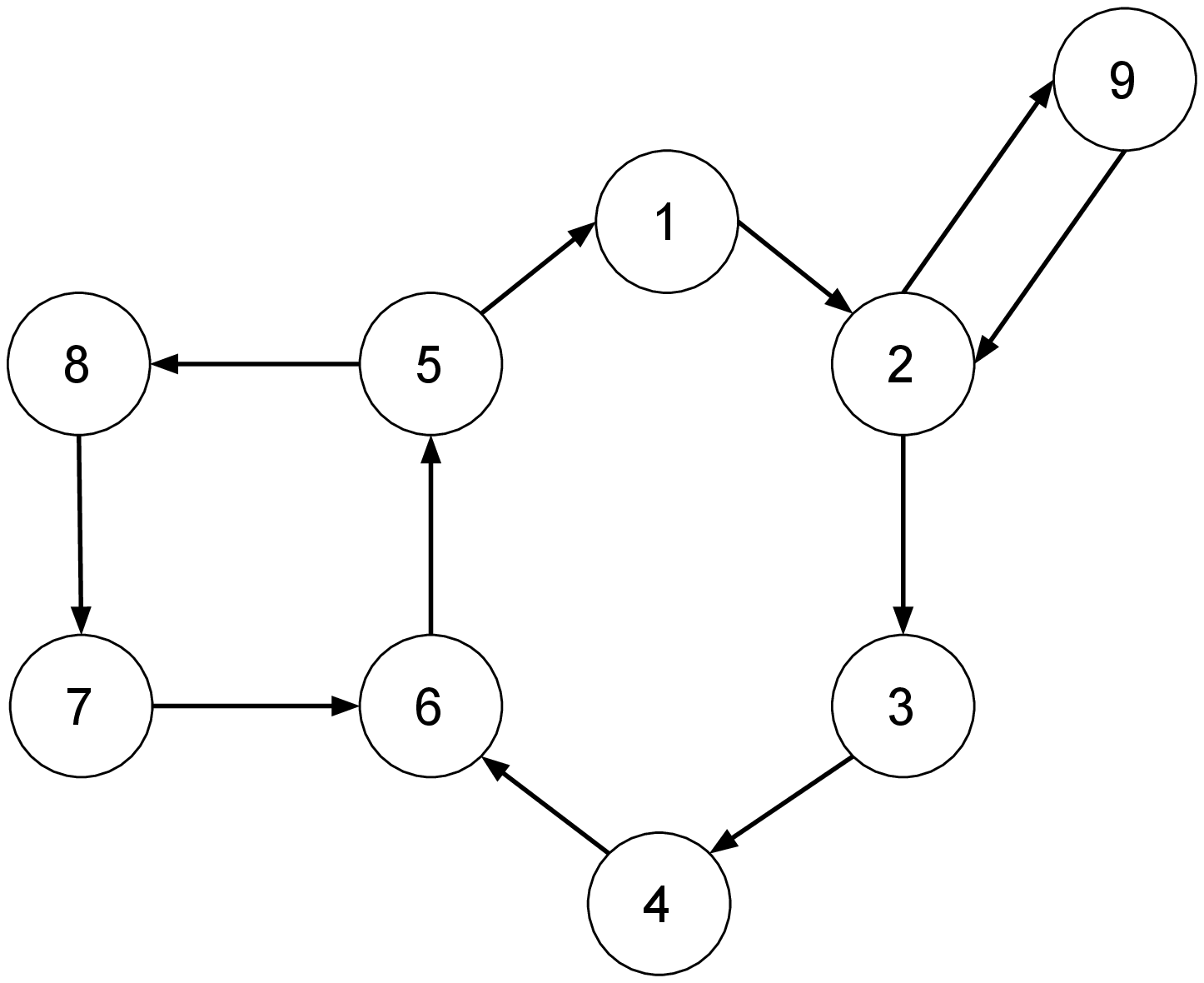}
				\caption{Figure illustrating two non-isolated MOCGs.}
				\label{non_iso}
			\end{subfigure}
			\begin{subfigure}{0.31\textwidth}
				\centering
				\hspace*{-8mm}
				\includegraphics[width=3in,height=2in,angle=0]{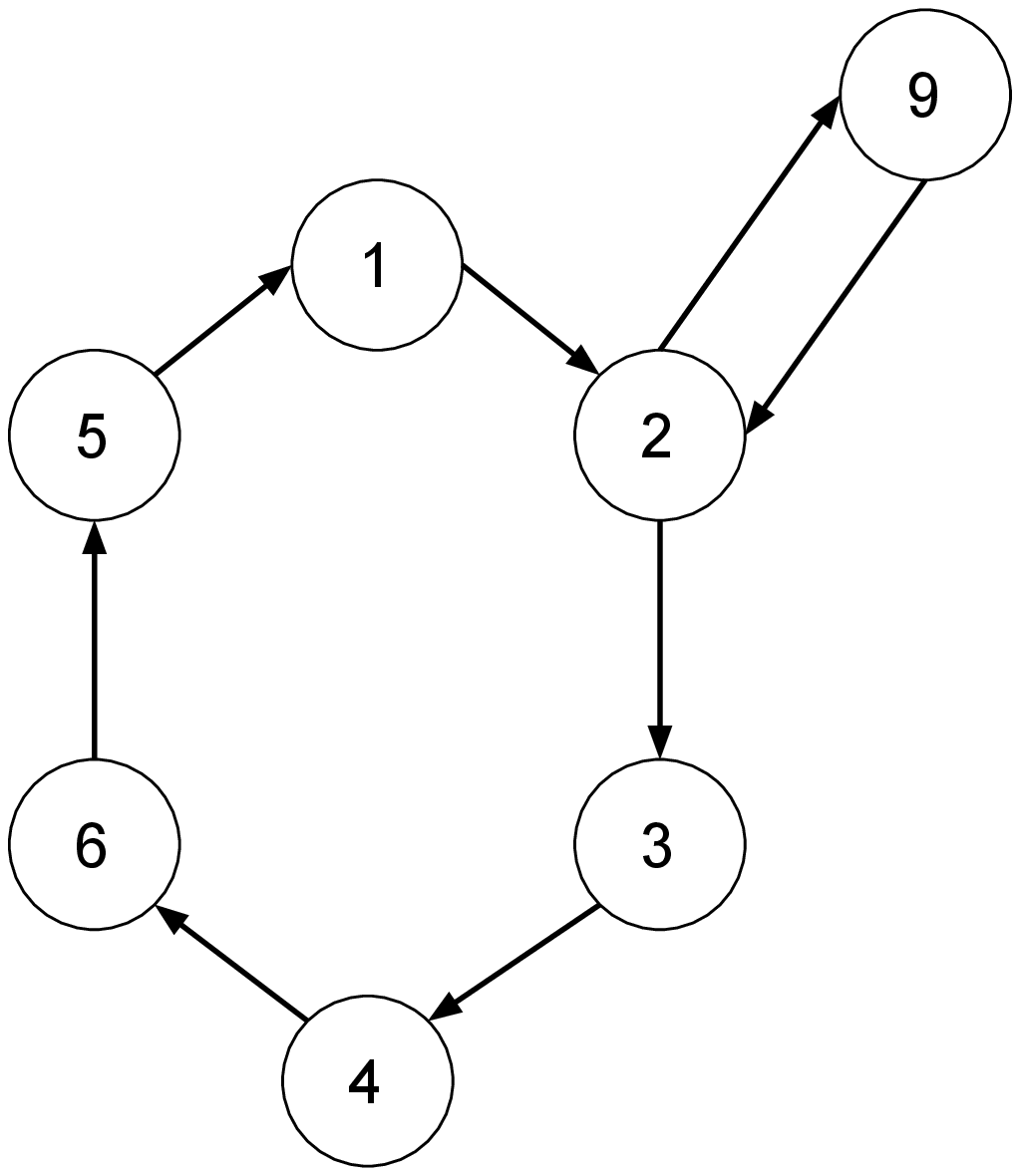}
				\caption{outer-cycle group of vertex $ 2 $.}
				\label{example3}
			\end{subfigure}
			\begin{subfigure}{0.31\textwidth}
				\centering
				\hspace*{-16mm}
				\includegraphics[width=3in,height=2in,angle=0]{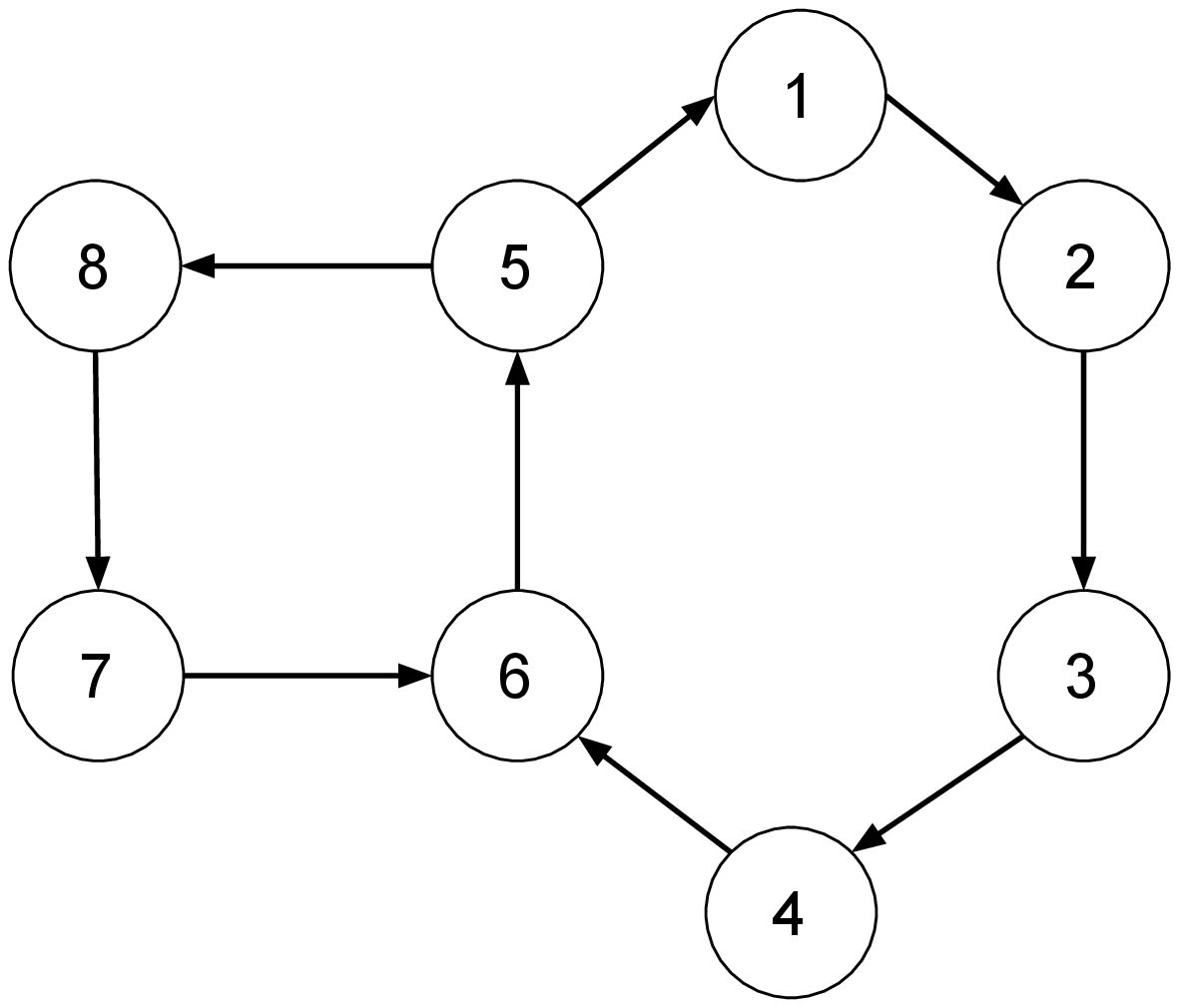}
				\caption{outer-cycle group of vertices $ 5$ and $6 $.}
				\label{example4}
			\end{subfigure}
			\caption{}
		\end{figure*}
		It can be observed that the outer cycle with vertex set $ \lbrace1,2,3,4,5,6\rbrace $ is common to both the MOCGs. Hence the two MOCGs, one with CCV $ 2 $ and one with CCV $ 6 $, are non-isolated MOCGs.
	\end{example}
\begin{lm}
	An MOCG will have a unique CCV.
\end{lm}
\begin{proof}
	MOCGs can be categorized into two types. MOCG which is an outer-cycle group of a unique non-inner vertex and MOCG which is an outer-cycle group of more than one non-inner vertices. In the former case CCV is unique.\\In the latter case, consider that an MOCG $ \mathcal{G} $ is union of $ n \geq 2 $ outer cycles and let $ m \geq2 $ non-inner vertices have $ \mathcal{G} $ as their outer-cycle group.
	\begin{claim}
		There is a subgraph $ \mathcal{S} $ of $ \mathcal{G} $ which is formed only by those $ m $ vertices and $ \mathcal{S} $ is a subgraph of all the $ n $ outer cycles.
	\end{claim}
\begin{proof}
	Let the $ m $ vertices be $ v_1$, $ v_2 $, $ \dots $, $ v_m $. Suppose $ \mathcal{S} $ does not contain the vertex $ v_m $. This means $ v_m $ is not present in at least one of the $ n $ outer cycles. Hence outer cycle group of $ v_m $ cannot consist of all those $ n $ outer cycles since $ v_m $ is not present in at least one of the $ n $ outer cycles. This is a contradiction.
	So all the $ m $ vertices are present in $ \mathcal{S} $.
	Now, if there is an $ m+1^{th} $ vertex $ v_{m+1} $ in $ \mathcal{S} $, then it means that $ v_{m+1} $ is present in all the $ n $ outer cycles and there would be $ m+1 $ non-inner vertices which have $ \mathcal{G} $ as their outer-cycle group. This is not possible since it is given that there are $ m $ non-vertices that have $ \mathcal{G} $ as their outer-cycle group.
\end{proof}
	\begin{claim}
		In $ \mathcal{G} $, the subgraph $ \mathcal{S} $ is always of the form 
		\begin{equation}
		\label{exp1}
		v_1 \rightarrow v_2 \rightarrow \dots \rightarrow v_{m-1}\rightarrow v_m
		\end{equation}, with appropriate labelling of the $ m $ vertices.
	\end{claim}
 	\begin{proof}
	Any acyclic structure except the structure in expression \ref{exp1} cannot be formed by the $ m $ vertices, since in that case, parallel paths will exist between at least one pair of vertices among the $ m $ vertices which implies that uniqueness of I-path will be violated.\\Any cyclic structure cannot be formed among the $ m $ vertices, since in that case, some of the $ m $ vertices will be present in more than $ n $ outer cycles which is not possible since each of the $ m $ vertices are present in only $ n $ outer cycles.
	Hence the only possibility of the form of $ \mathcal{S} $ is as shown in expression \ref{exp1}.
  	\end{proof}
   Since the form of $ \mathcal{S} $ is fixed as in expression \ref{exp1} (with the assumption of appropriate labelling), and since $ \mathcal{S} $ is present in all the $ n $ outer cycles, $ v_1 $ should have in-degree equal to $ n $ and all the other vertices $ v_2 $, $ v_3 $, $ \dots $, $ v_m $ will have in-degree equal to $ 1 $. Since $ n\geq 2 $, $ v_1 $ will always have the highest in-degree and hence $ v_1 $ will be the CCV of $ \mathcal{G} $. Thus a unique CCV exists for an MOCG.
\end{proof}
	\subsection{\textit{Construction} $ 2 $}
	Let the given IC structure with outer cycles be $ \mathcal{G} $.
 Let\\
 $ V_I=\lbrace1,2,\dots,K\rbrace $ be the set of inner vertices,\\ $ V_{NI}=\lbrace K+1,K+2,\dots,N\rbrace $ be the set of non-inner vertices,\\
$ V_{OC} $ be the union of the vertex sets corresponding to all of the outer cycles present in $ \mathcal{G} $,\\ $ V_{MOCG} $ be the union of vertex sets of outer cycles present in any of the MOCGs,\\$ V_{OC}(j) $ denote the union of vertex sets corresponding to the outer cycles in the outer-cycle group of $ j $ (i.e., vertices of those outer cycles which contain $ j $),\\$ N^+_C(j) $ denote the set of vertices in $ V_{OC}(j) \cap N^+_{\mathcal{G}}(j) $ (it is the out-neighborhood of $ j $ in its outer-cycle group)\\$ V_{OCGI} $ be the union of all the isolated MOCGs.\\ \\
	The index code construction is as follows.
	\begin{cons}
		\label{cons2}
		\begin{enumerate}
\item An index code symbol $ W_I $ obtained by XOR of messages corresponding to inner vertices is transmitted, where 
			\begin{displaymath}
			W_I=\underset{i=1}{\overset{K}{\bigoplus}}x_i.
			\end{displaymath}
\item For non-inner vertices not in any outer cycle, an index code symbol corresponding to each such non-inner vertex, obtained by XOR of message corresponding to the non-inner vertex with the messages corresponding to the vertices in the out-neighborhood excluding non-inner vertices that are part of outer cycles of the non-inner vertex is transmitted, i.e., for $ j \in V_{NI}\backslash V_{OC} $,
\begin{displaymath} 
			W_j = x_j \underset{q\in N^+_{\mathcal{G}}(j)\backslash V_{OC}}{\bigoplus} x_q. 
\end{displaymath}
\item Consider the isolated MOCGs and non-isolated MOCGs which are union of even number of outer cycles.	For each such MOCG
			\begin{itemize}
				\item Let the CCV of the MOCG be $ j_1 $ and let $ j $ be an arbitrary element in pre-central cycle vertex set of $ j_1 $.
\item For the vertex $ j $, $ W_j $ is transmitted, where
				\begin{displaymath}
				W_{j} = x_{j} \underset{q \in \lbrace N^+_{\mathcal{G}}(j) \backslash V_{OC}\rbrace \cup \lbrace N^+_C(j) \backslash \lbrace j_1\rbrace \rbrace}{\bigoplus}x_q.
				\end{displaymath}
			\end{itemize}	
\item Consider the non-isolated MOCGs which are union of odd number of outer cycles. For each such MOCG
				\begin{itemize}
\item Let the CCV of the MOCGs be $ j_1 $. 
\item Let $ PV(j_1) $ be the set of two arbitrary elements in pre-central cycle vertex set of $ j_1 $.
\item For each $j \in PV(j_1) $, $ W_j $ is transmitted, where
				\begin{displaymath}
				W_j= x_{j} \underset{q \in \lbrace N^+_{\mathcal{G}}(j) \backslash V_{OC}\rbrace \cup \lbrace N^+_C(j) \backslash \lbrace j_1\rbrace \rbrace}{\bigoplus}x_q.
				\end{displaymath}
\end{itemize}
\item For all other non-inner vertices $ j $, $ W_j $ is transmitted, where
			\begin{displaymath}
			W_j =  x_j \underset{q\in \lbrace N^+_{\mathcal{G}}(j) \backslash V_{OC} \rbrace \cup N^+_{C}(j)}{\bigoplus}x_q.
			\end{displaymath}
\end{enumerate}

		Since an index code symbol is transmitted for every non-inner vertex and one index code symbol is transmitted for all the inner vertices, the length of the index code obtained is $ N-K+1 $.
\end{cons}
\begin{ex}
\label{ex_1}
\begin{figure}
			\centering
			\includegraphics[width=\columnwidth,height=\columnwidth,angle=0]{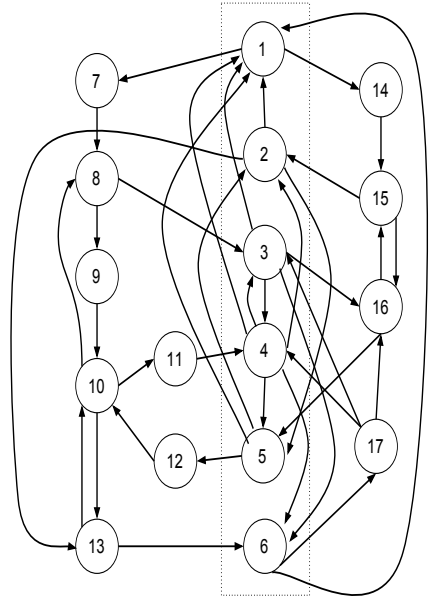}
			\caption{$ 6 $-IC structure $ \mathcal{G}_1 $ with outer cycles.}
			\label{fig1}
		\end{figure}
		\begin{figure*}[!t]
			\centering
			\begin{subfigure}{.31\textwidth}
				\centering
				\includegraphics[width=15pc]{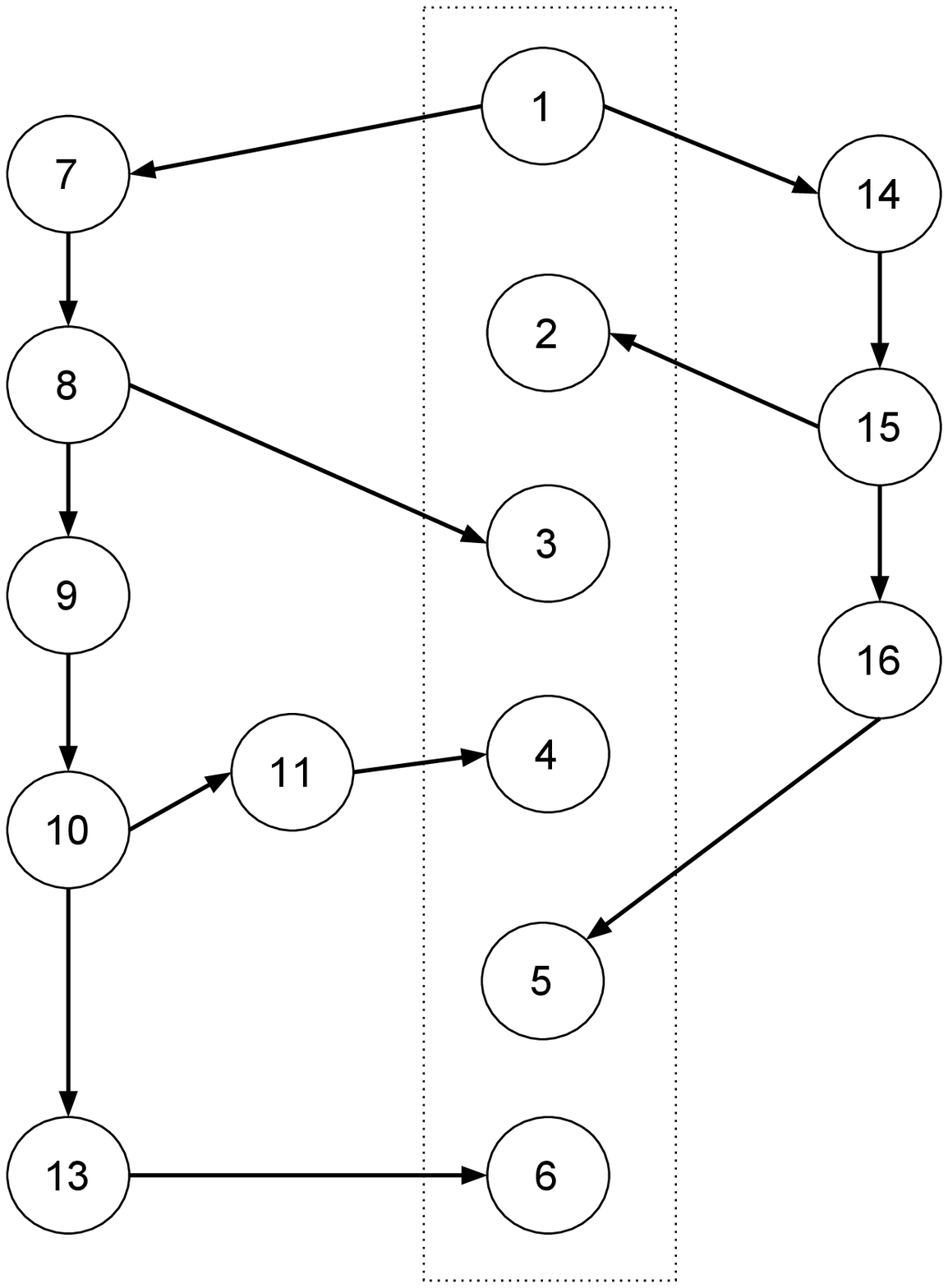}
				\caption{}
				\label{rt11}
			\end{subfigure}%
			\begin{subfigure}{.31\textwidth}
				\centering
				\includegraphics[width=15pc]{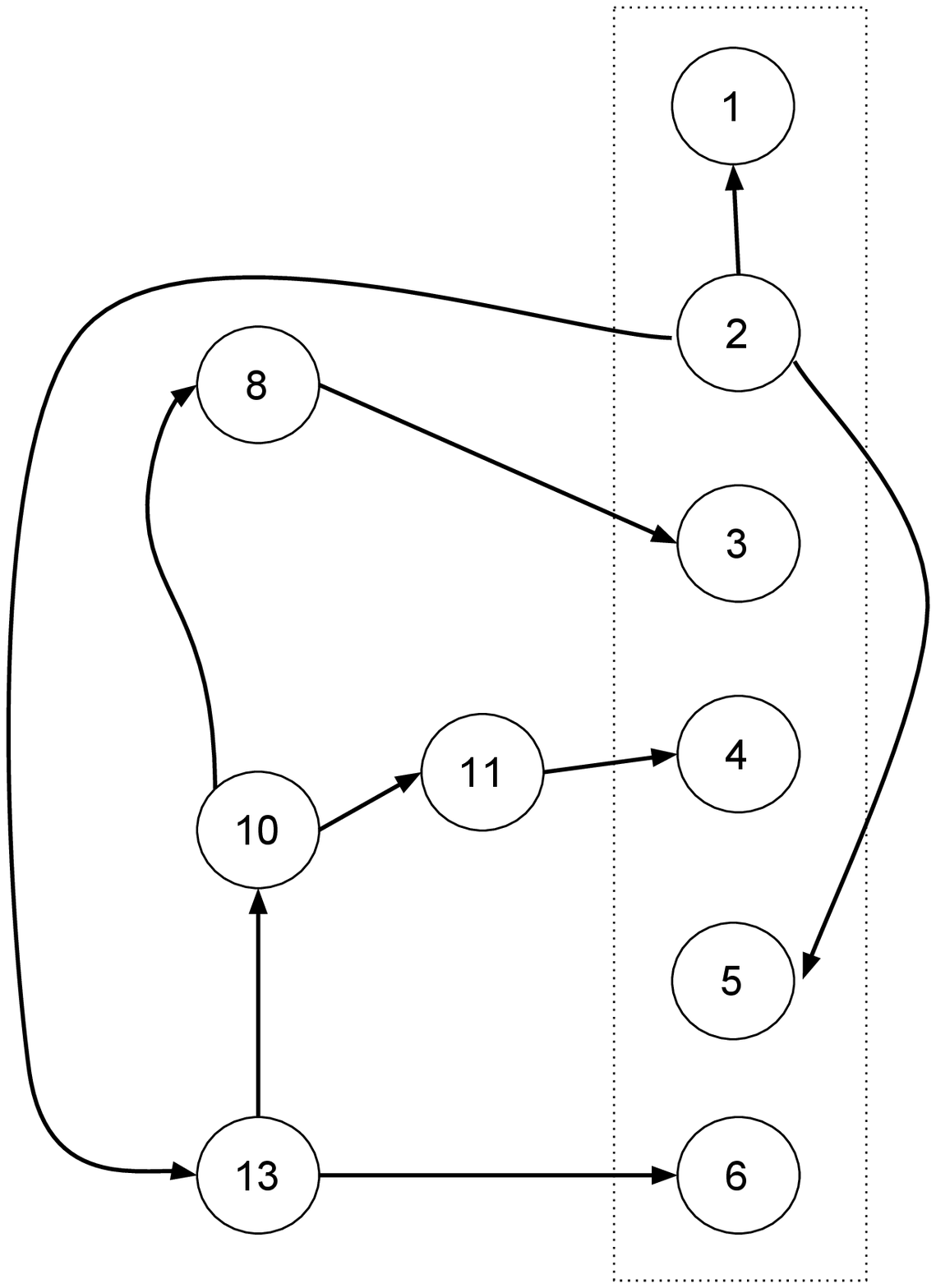}
				\caption{}
				\label{rt12}
			\end{subfigure}
			\begin{subfigure}{.31\textwidth}
				\centering
				\includegraphics[width=15pc]{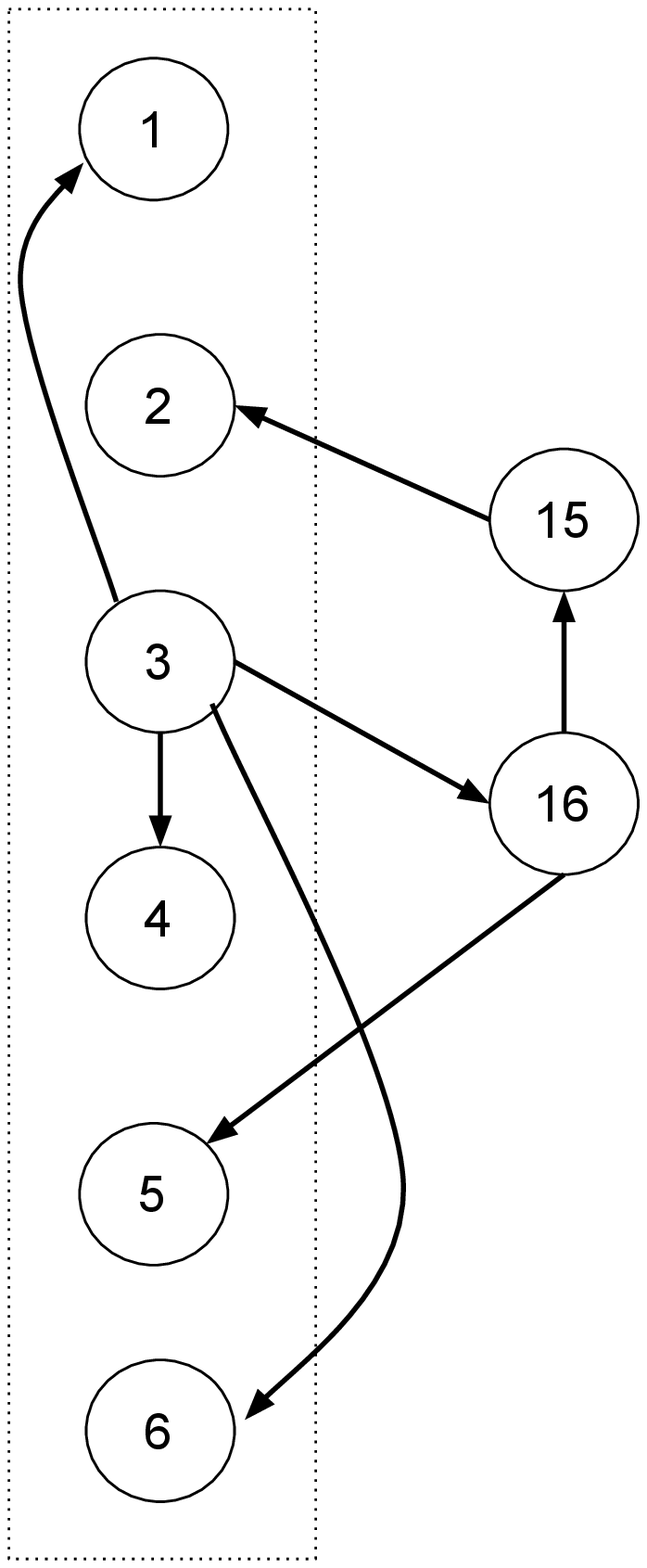}
				\caption{}
				\label{rt13}
			\end{subfigure}
			\begin{subfigure}{.31\textwidth}
				\centering
				\includegraphics[width=15pc]{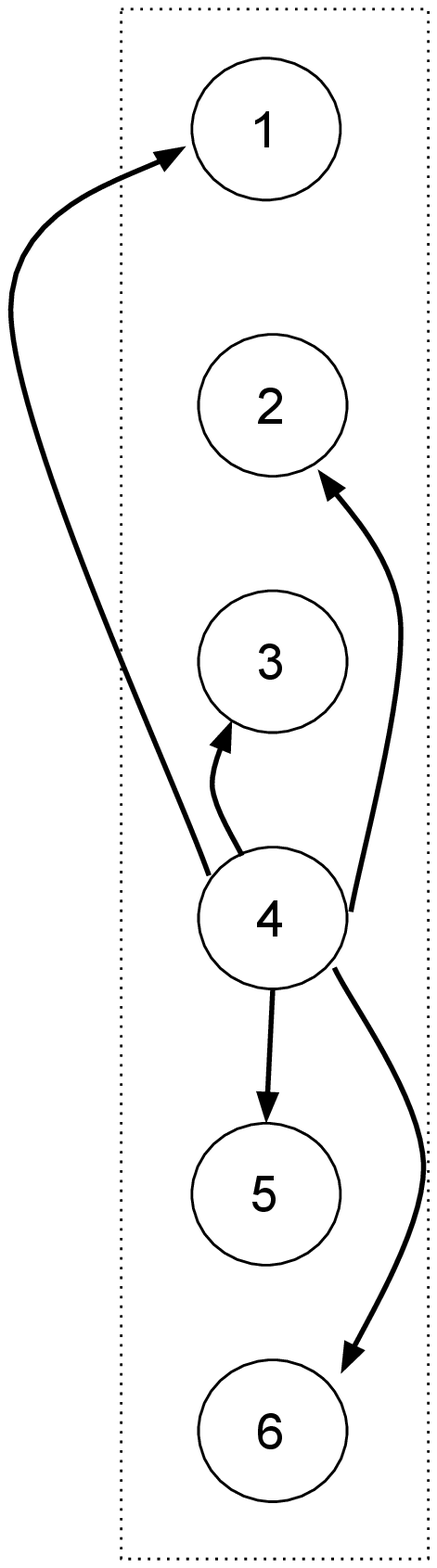}
				\caption{}
				\label{rt14}
			\end{subfigure}%
			\begin{subfigure}{.31\textwidth}
				\centering
				\includegraphics[width=15pc]{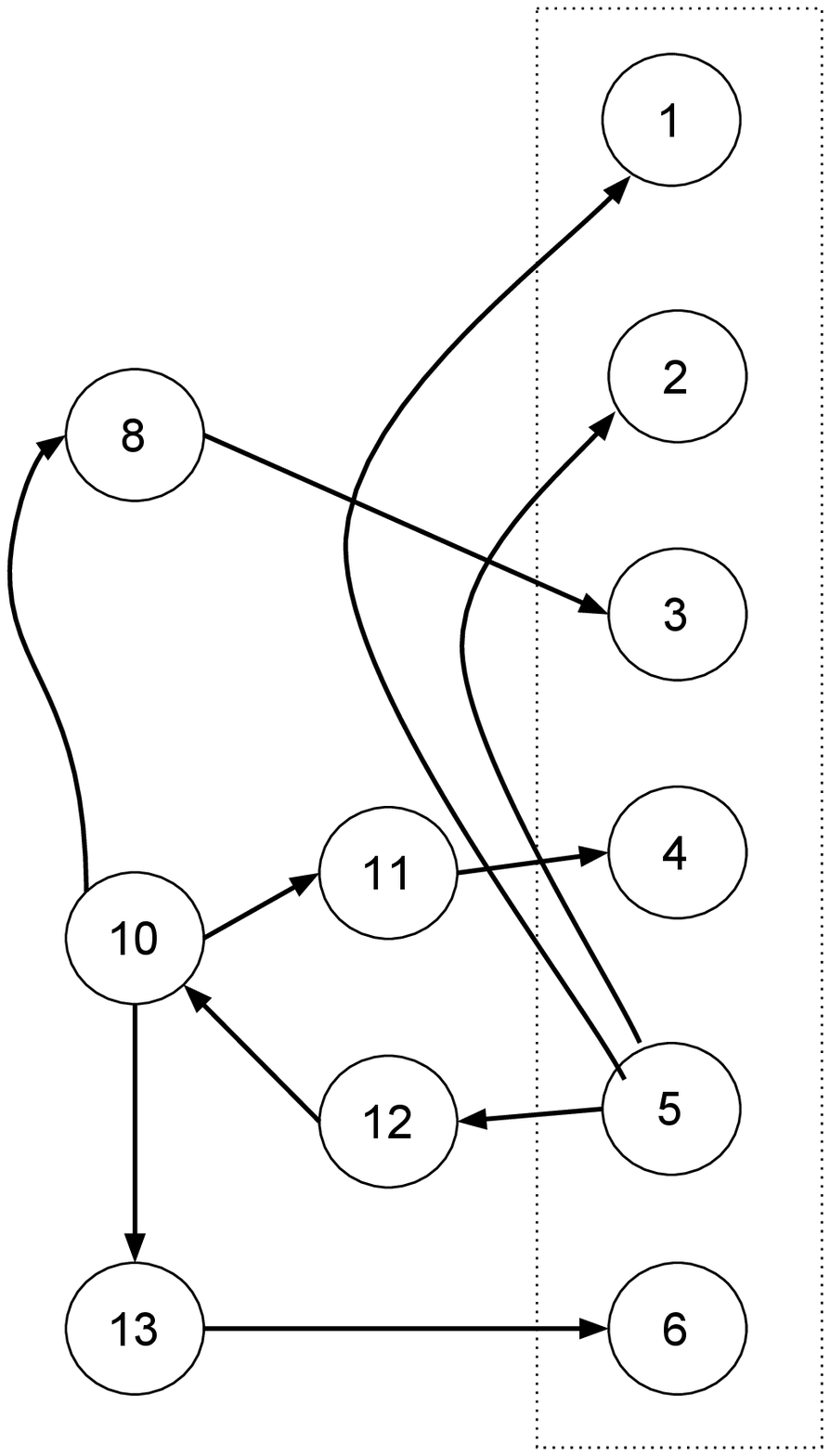}
				\caption{}
				\label{rt15}
			\end{subfigure}
			\begin{subfigure}{.31\textwidth}
				\centering
				\includegraphics[width=15pc]{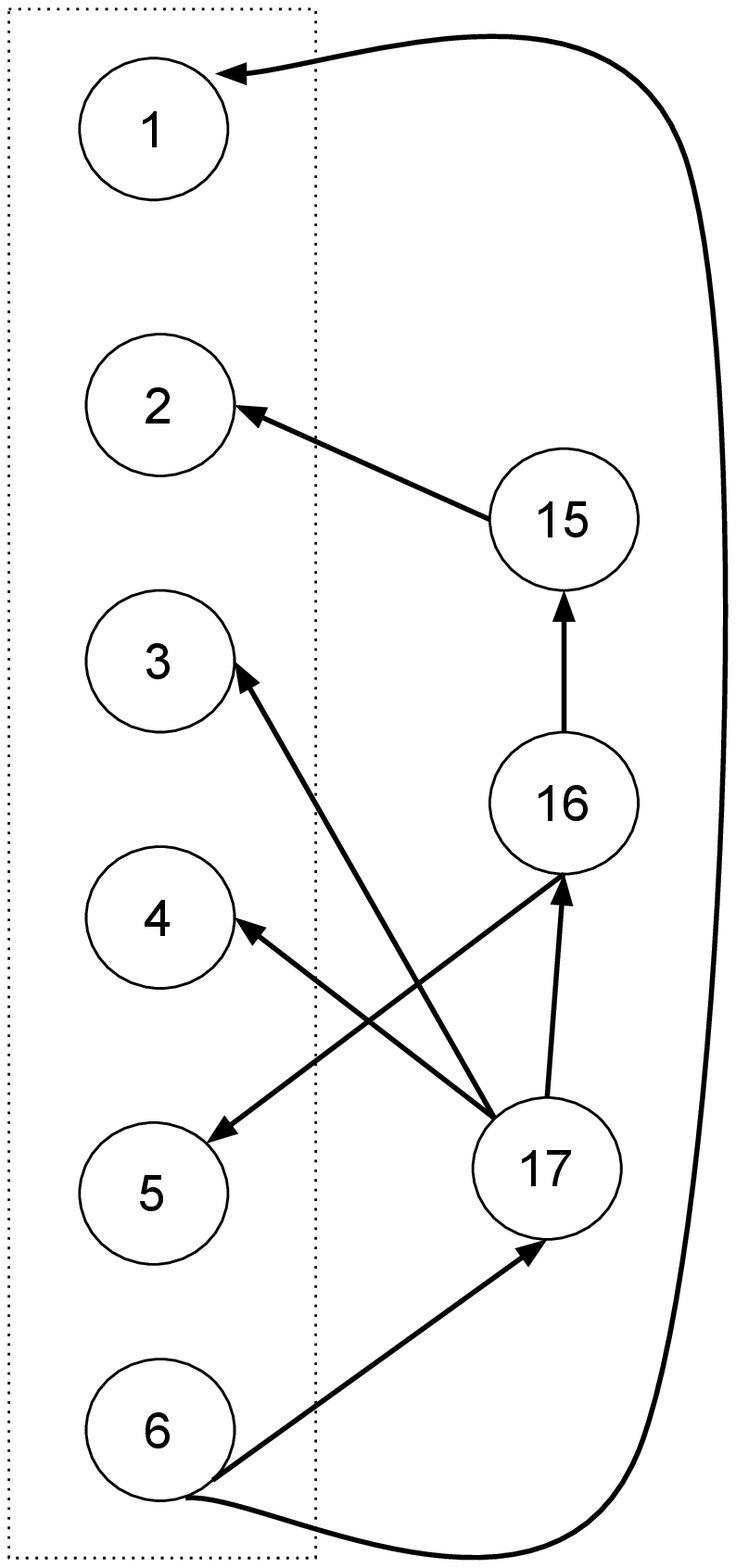}
				\caption{}
				\label{rt16}
			\end{subfigure}
			\caption{Figures showing rooted trees of inner vertices $ 1,2,3,4,5,6 $ of $ \mathcal{G}_1 $, respectively.}
		\end{figure*}
		Consider $ \mathcal{G}_1 $, a side information graph which is a $ 6 $-IC structure with inner vertex set $ V_I = \lbrace1,2,3,4,5,6\rbrace $ given in Fig. \ref{fig1}. It can be verified that 
\begin{enumerate}
			\item there are no cycles containing only one vertex from the set $\lbrace1,2,3,4,5,6\rbrace$ in $\mathcal{G}_1$ (i.e., no I-cycles),
			\item using the rooted trees for each vertex in the set $\lbrace1,2,3,4,5,6\rbrace$, which are given in Fig. \ref{rt11}, \ref{rt12}, \ref{rt13}, \ref{rt14}, \ref{rt15} and \ref{rt16} respectively, there exists a unique path between any two different vertices in $ V_I $ in $ \mathcal{G}_1 $ and does not contain any other vertex in $ V_I $ (i.e., unique I-path between any pair of inner vertices),
			\item $\mathcal{G}_1$ is the union of all the $6$ rooted trees.
		\end{enumerate}
		Now, we have\\
		$ V_I = \lbrace 1,2,3,4,5,6 \rbrace. $\\
$ 		V_{NI} =  \lbrace7,8,9,10,11,12,13,14,15,16,17\rbrace $.\\
$ 	V_{OC}=\lbrace 8,9,10,13,15,16 \rbrace $.\\
$ 		V_{MOCG}=\lbrace 8,9,10,13\rbrace $.\\
		 There is one isolated MOCG with CCV $ 10 $, formed by the union of two outer cycles with vertex sets $ \lbrace 8,9,10 \rbrace $ and $ \lbrace 10,13 \rbrace$. Since there are no other MOCGs, $ V_{OCGI}=V_{MOCG} $. \textit{Construction} \ref{cons2} is carried out as follows.
\begin{enumerate}
			\item $ W_I= x_1\oplus x_2\oplus x_3\oplus x_4\oplus x_5\oplus x_6 $.
			\item $ V_{NI}\backslash V_{OC} = \lbrace7,11,12,14,17\rbrace $.
			\begin{itemize}
				\item $ j=7 $ and $ N^+_{\mathcal{G}_1}(7)\backslash V_{OC}= \phi $. Hence $ W_7=x_7 $.
				\item $ j=11 $ and $ N^+_{\mathcal{G}_1}(11)\backslash V_{OC}= \lbrace4\rbrace $. Hence $ W_{11}=x_{11}\oplus x_4 $.
				\item $ j=12 $ and $ N^+_{\mathcal{G}_1}(12)\backslash V_{OC}= \phi $. Hence $ W_{12}=x_{12} $.
				\item $ j=14 $ and $ N^+_{\mathcal{G}_1}(14)\backslash V_{OC}= \phi $. Hence $ W_{14}=x_{14} $.
				\item $ j=17 $ and $ N^+_{\mathcal{G}_1}(17)\backslash V_{OC}= \lbrace3,4\rbrace $. Hence $ W_{17}=x_{17}\oplus x_3\oplus x_4 $.
			\end{itemize}
			\item The isolated MOCG with CCV $ 10 $ is union of even number of cycle sets.
			\begin{itemize}
				\item $ j_1=10 $. $ j $ can be $ 9 $ or $ 13 $. Let $ j=9 $ (The case of $ j=13 $ is considered in the end).\\
				$ N^+_{\mathcal{G}}(9) \backslash V_{OC}  = \phi$.\\
				$ N^+_{C}(9)=\lbrace10\rbrace $. Hence\\ $  \lbrace N^+_{\mathcal{G}}(9) \backslash V_{OC}\rbrace \cup \lbrace N^+_{C}(9) \backslash  \lbrace 10\rbrace \rbrace = \phi $. So,\\ $ W_9=x_9 $.
			\end{itemize}
			\item Step $ 4 $ can be skipped as there are no non-isolated MOCGs.
			\item The remaining non-inner vertices are $ \lbrace 8,10,13,15,16 \rbrace $.
			\begin{itemize}
				\item $ j=8 $. $  N^+_{\mathcal{G}}(8) \backslash V_{OC}  = \lbrace3\rbrace $ and $ N^+_{C}(8) = \lbrace 9\rbrace $ and hence $ \lbrace N^+_{\mathcal{G}}(8) \backslash V_{OC} \rbrace \cup N^+_{C}(8)= \lbrace3,9\rbrace $. So, $ W_8=x_8\oplus x_3\oplus x_9 $.
				\item $ j=10 $. $  N^+_{\mathcal{G}}(10) \backslash V_{OC}  =\lbrace 11\rbrace $ and $ N^+_{C}(10) = \lbrace 8,13 \rbrace $ and hence $ \lbrace N^+_{\mathcal{G}}(10) \backslash V_{OC} \rbrace \cup N^+_{C}(10)= \lbrace8,11,13\rbrace $. So, $ W_{10}=x_{10}\oplus x_8\oplus x_{11}\oplus x_{13} $.
				\item $ j=13 $. $  N^+_{\mathcal{G}}(13) \backslash V_{OC}  = \lbrace 6 \rbrace$ and $ N^+_{C}(13) = \lbrace 10\rbrace $ and hence $ \lbrace N^+_{\mathcal{G}}(13) \backslash V_{OC} \rbrace \cup N^+_{C}(13)= \lbrace6,10\rbrace $. So, $ W_{13}=x_{13}\oplus x_6\oplus x_{10} $.
				\item $ j=15 $. $  N^+_{\mathcal{G}}(15) \backslash V_{OC}  = \lbrace 2 \rbrace $ and $ N^+_{C}(15) = \lbrace 16\rbrace $ and hence $ \lbrace N^+_{\mathcal{G}}(15) \backslash V_{OC} \rbrace \cup N^+_{C}(15)= \lbrace2,16\rbrace $. So, $ W_{15}=x_{15}\oplus x_2\oplus x_{16} $.
				\item $ j=16 $. $  N^+_{\mathcal{G}}(16) \backslash V_{OC}  = \lbrace 5 \rbrace $ and $ N^+_{C}(16) = \lbrace 15\rbrace $ and hence $ \lbrace N^+_{\mathcal{G}}(16) \backslash V_{OC} \rbrace \cup N^+_{C}(16)= \lbrace5,15\rbrace $. So, $ W_{16}=x_{16}\oplus x_5\oplus x_{15} $.
			\end{itemize}
\end{enumerate}

The obtained index code symbols are listed as follows.\\
\hspace*{1.0cm} $ 			 W_I= x_1\oplus x_2\oplus x_3\oplus x_4\oplus x_5\oplus x_6. $\\
\hspace*{1.0cm} $ 			 W_7=x_7 . $\\
\hspace*{1.0cm} $ 			 W_8=x_8\oplus x_3\oplus x_9 . $\\
\hspace*{1.0cm} $ 			 W_9=x_9 . $\\
\hspace*{1.0cm} $ 			 W_{10}=x_{10}\oplus x_8\oplus x_{11}\oplus x_{13} . $\\
\hspace*{1.0cm} $ 			 W_{11}=x_{11}\oplus x_4 . $\\
\hspace*{1.0cm} $ 			 W_{12}=x_{12} . $\\
\hspace*{1.0cm} $ 			 W_{13}=x_{13}\oplus x_6\oplus x_{10} . $\\
\hspace*{1.0cm} $ 			 W_{14}=x_{14} . $\\
\hspace*{1.0cm} $ 			 W_{15}=x_{15}\oplus x_2\oplus x_{16} . $\\
\hspace*{1.0cm} $ 			 W_{16}=x_{16}\oplus x_5\oplus x_{15} . $\\
\hspace*{1.0cm} $ 			 W_{17}=x_{17}\oplus x_3\oplus x_4 . $\\ \\
	If in step $ 3 $ where $ j_1=10 $, if $ j=13 $, index code symbols $ W_9 $ and $ W_{13} $ will be\\
\hspace*{1.0cm} $ 		W_9 = x_9\oplus x_{10}. $\\
\hspace*{1.0cm} $ 		W_{13} = x_{13}\oplus x_6 $\\
 and all the remaining index code symbols will remain same.\\
\end{ex}

\subsection{Algorithm $ 2 $}
	Let the IC structure with outer cycles be called $ \mathcal{G} $. Consider the following algorithm to decode messages corresponding to inner vertices.
	\begin{algo}
		\label{algo2}
		Let the inner vertex be $ i $. Let the subset of vertices in $ V_{NI}(i) $ which are part of outer cycles be called $ V_{NIC}(i) $. Now, define the set $ V'_{NI}(i) $ as
		\begin{displaymath}
		V'_{NI}(i)=  \underset{j \in V_{NIC}(i)}{\bigcup} V_{OC}(j) 
		\end{displaymath} 
		Compute $ Z_i $ as 
		\begin{equation*}
		Z_i = W_I \underset{q\in V_{NI}(i) \backslash V_{NIC}(i)}{\bigoplus}W_q \underset{j\in V'_{NI}(i)}{\bigoplus}W_j 
		\end{equation*}to decode $ x_i $ for $ i \in \lbrace1,2,\dots,K\rbrace $.
		The messages corresponding to non-inner vertices are decoded directly by using the index code symbols corresponding to the respective non-inner vertices.
\end{algo}
\begin{remark}
		In a rooted tree $ T_i $ in which \textit{c}$ 2 $ is violated, \textit{Algorithm} \ref{algo2} includes the non-inner vertices that are not present in the rooted tree but are in the out-neighborhood of some vertices in $ V_{NI}(i) $ to compute $ Z_i $. (This is the only difference in \textit{Algorithm} $ \ref{algo2} $  compared to \textit{Algorithm} \ref{algo1}).
\end{remark}
\begin{exmp}[continued]
		The decoding is done for messages corresponding to non-inner vertices using respective index code symbols. The decoding of messages corresponding to inner vertices is done using \textit{Algorithm} $ 2 $ as follows.\\
		\begin{itemize}
			\item $ i=1 $. $ V_{NIC}(1)=\lbrace 8,9,10,13,15,16\rbrace $. 
			Hence $ V'_{NI}(1)=\lbrace8,9,10,13,15,16\rbrace $ and $ V_{NI}(1)\backslash V_{NIC}(1) =\lbrace 7,11,14\rbrace $. 
			$ Z_1 = W_I\oplus W_7\oplus W_8\oplus W_9\oplus W_{10}\oplus W_{11}\oplus W_{13}\oplus W_{14}\oplus W_{15}\oplus W_{16} $ which results in $ Z_1=x_1\oplus x_7\oplus x_{14} $ and user $ 1 $ has $ x_7 $ and $ x_{14} $ in its side information.
			\item $ i=2 $. $ V_{NIC}(2)=\lbrace 8,10,13\rbrace $. 
			Hence $ V'_{NI}(2)=\lbrace8,9,10,13\rbrace $ and $ V_{NI}(2)\backslash V_{NIC}(2) =\lbrace 11\rbrace $. 
			$ Z_2 = W_I\oplus W_8\oplus W_9\oplus W_{10}\oplus W_{11}\oplus W_{13} $ which results in $ Z_2=x_1\oplus x_2\oplus x_{5} $ and user $ 2 $ has $ x_1 $ and $ x_{5} $ in its side information.
			\item $ i=3 $. $ V_{NIC}(3)=\lbrace 15,16\rbrace $. 
			Hence $ V'_{NI}(3)=\lbrace 15,16\rbrace $ and $ V_{NI}(i)\backslash V_{NIC}(i) =\phi $. 
			$ Z_3 = W_I\oplus W_{15}\oplus W_{16} $ which results in $ Z_3=x_3\oplus x_1\oplus x_{4}\oplus x_6 $ and user $ 3 $ has $ x_1 $, $ x_4 $ and $ x_6 $ in its side information.
			\item $ i=4 $. $ V_{NIC}(4)=\phi $. 
			Hence $ V'_{NI}(4)=\phi $ and $ V_{NI}(4)\backslash V_{NIC}(4) =\phi $. 
			$ Z_4 = W_I = x_1\oplus x_2\oplus x_3\oplus x_4\oplus x_5\oplus x_6 $ and user $ 4 $ has $ x_1 $, $ x_2 $, $ x_3 $, $ x_5 $ and $ x_{6} $ in its side information.
			\item $ i=5 $. $ V_{NIC}(5)=\lbrace 8,10,13\rbrace $. 
			Hence $ V'_{NI}(5)=\lbrace8,9,10,13\rbrace $ and $ V_{NI}(5)\backslash V_{NIC}(5) =\lbrace 11,12\rbrace $. 
			$ Z_5 = W_I\oplus W_8\oplus W_9\oplus W_{10}\oplus W_{11}\oplus W_{12}\oplus W_{13} $ which results in $ Z_5=x_5\oplus x_1\oplus x_2\oplus x_{12} $ and user $ 5 $ has $ x_1 $, $ x_2 $ and $ x_{12} $ in its side information.
			\item $ i=6 $. $ V_{NIC}(6)=\lbrace 15,16\rbrace $. 
			Hence $ V'_{NI}(6)=\lbrace 15,16\rbrace $ and $ V_{NI}(6)\backslash V_{NIC}(6) =\lbrace 17\rbrace $. 
			$ Z_6 = W_I\oplus W_{15}\oplus W_{16}\oplus W_{17} $ which results in $ Z_6=x_6\oplus x_1\oplus x_{17} $ and user $ 6 $ has $ x_1 $ and $ x_{17} $ in its side information.
		\end{itemize}
\end{exmp}
	In the following theorem, it is shown that an index code obtained by using \textit{Construction} \ref{cons2} on an IC structure with outer cycles is decodable using \textit{Algorithm} \ref{algo2}.
\subsection{The Main Result}
\begin{thm}
\label{thm1}
		An index code obtained by using \textit{Construction} $ 2 $ on an IC structure with outer cycles, $ \mathcal{G} $, is decodable using \textit{Algorithm} $ 2 $.
	\end{thm}
	\begin{proof}IC structures with outer cycles can be categorized into two classes. One class is of those IC structures with outer cycles that satisfy \textit{c}$ 2 $ and the other class is of those IC structures that do not satisfy \textit{c}$ 2 $. Consider an inner vertex $ i $ and its rooted tree $ T_i $. It shall be shown that the inner and non-inner vertices at depth $ \geq 2 $ in $ T_i $ will be cancelled in the computation of $ Z_i $ irrespective of whether \textit{c}$ 2 $ is satisfied or not. The messages corresponding to the non-inner vertices can be decoded directly by using the index code symbols corresponding to the respective non-inner vertices. The different type of vertices that can be present at depth $ \geq2 $ in a rooted tree $ T_i $ are listed as follows.
		\begin{itemize}
			\item (Type $ 1 $ vertices) Inner vertices.
			\item (Type $ 2 $ vertices) Non-inner vertices that are not present in any outer cycle.
			\item (Type $ 3 $ vertices) CCVs corresponding to non-isolated and isolated MOCGs which are union of even number of outer cycles.
			\item (Type $ 4 $ vertices) CCVs corresponding to non- isolated MOCGs which are union of odd number of outer cycles.
			\item (Type $ 5 $ vertices) Vertices present in outer cycles that are not part of any MOCG, vertices which are not CCVs but present in MOCGs and CCVs corresponding to isolated MOCGs which are union of even number of outer cycles.
		\end{itemize}
		The message corresponding to a type $ 1 $ vertex appears exactly twice in $ Z_i $, once in $ W_I $ and once in the index code symbol corresponding to its immediate predecessor in $ T_i $, and is, hence, cancelled.\\
		The message corresponding to a type $ 2 $ vertex appears exactly twice in $ Z_i $, once in index code symbol corresponding to itself and once in index code symbol corresponding to its immediate predecessor in the rooted tree, and is, hence, cancelled.\\
		The message corresponding to a type $ 3 $ vertex appears even number of times in $ Z_i $, once in index code symbol corresponding to itself and odd number of times in index code symbols corresponding to its pre-central cycle vertices (as it is not encoded in the index code symbol of one pre-central cycle vertex, in step $ 3 $ of \textit{Construction} \ref{cons2}), and is, hence, cancelled.\\
		The message corresponding to a type $ 4 $ vertex appears even number of times in $ Z_i $, once in index code symbol corresponding to itself and odd number of times in index code symbols corresponding to its pre-central cycle vertices (as it is not encoded in the index code symbols of two pre-central cycle vertices, in step $ 4 $ of \textit{Construction} \ref{cons2}), and is, hence, cancelled.\\
		Among messages corresponding to type $ 5 $ vertices, the messages corresponding to vertices present in outer cycles that are not part of any MOCG and vertices which are not CCVs but present in MOCGs appear exactly twice, once in index code symbol corresponding to their immediate predecessor in the outer cycles and MOCG, respectively, and once in index code symbols corresponding to themselves and  are, hence, cancelled. A message corresponding to the CCV of an isolated MOCG appears even number of times in $ Z_i $, once in index code symbol corresponding to itself and in all of the index code symbols corresponding to its pre-central cycle vertices (which are odd in number), and is, hence, cancelled.
		\begin{remark}
			If any type $ k $ vertex, $ j $, where $ k\in \lbrace3,4,5\rbrace $, is present in $ T_i $ at a depth $ \geq2 $, then all the vertices in $ V_{OC}(j) $ are included in $ Z_i $ (if \textit{c}$ 2 $ is violated in $ T_i $, these are included as $ V'_{NI}(i) $ in $ Z_i $ and if \textit{c}$ 2 $ is not violated in $ T_i $, then the vertices in $ V_{OC}(j) $ are present in $ V_{NI}(i) $ and are hence included in $ Z_i $) due to which cancellation of the message corresponding to that type $ k $ vertex happens. Also, in \textit{Construction} \ref{cons2}, a message corresponding to a non-inner vertex $ q $, in an outer cycle is not encoded into index code symbol corresponding to a non-inner vertex which is not present in outer-cycle group of $ q $.
		\end{remark}
		Thus an index code obtained by using \textit{Construction} \ref{cons2} on an IC structure with outer cycles, $ \mathcal{G} $, is decodable using \textit{Algorithm} \ref{algo2}.
\end{proof}
\section{Illustrating Example}
	\label{sec_ex}In this section, \textit{Construction} \ref{cons2} and \textit{Algorithm} \ref{algo2} are illustrated for an IC structure with outer cycles which has both isolated and non-isolated MOCGs. The previous example had only an isolated MOCG.
\begin{ex}
		Consider $ \mathcal{G}$, a side information graph which is a $ 10 $-IC structure with outer cycles and with inner vertex set $ V_I = \lbrace1,2,3,4,5,6,7,8,9,10\rbrace $ given in Fig. \ref{fig2}.
		\begin{figure*}[!t]
			\centering
			\includegraphics[scale=0.6]{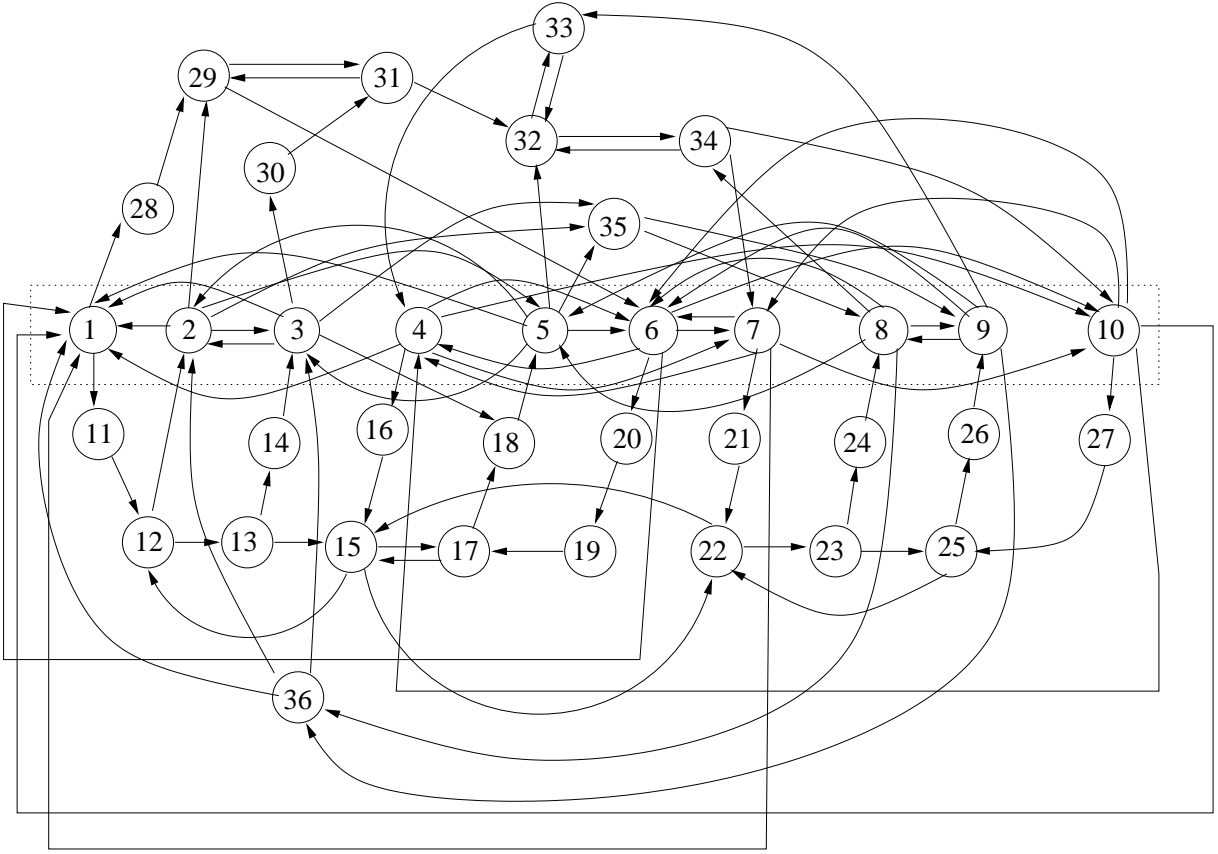}
			\caption{$ 10 $-IC structure with outer cycles, $ \mathcal{G} $.}
			\label{fig2}
		\end{figure*} It can easily be verified that 
		\begin{enumerate}
			\item there are no cycles containing only one vertex from the set $\lbrace1,2,3,4,5,6,7,8,9,10\rbrace$ in $\mathcal{G}$ (i.e., no I-cycles),
			\item using the rooted trees for each vertex in the set $\lbrace1,2,3,4,5,6,7,8,9,10\rbrace$, which are given in Fig. \ref{rt21}, \ref{rt22}, \ref{rt23}, \ref{rt24}, \ref{rt25}, \ref{rt26}, \ref{rt27}, \ref{rt28}, \ref{rt29} and \ref{rt210} respectively, there exists a unique path between any two different vertices in $ V_I $ in $ \mathcal{G} $ and does not contain any other vertex in $ V_I $ (i.e., unique I-path between any pair of inner vertices),
			\item $\mathcal{G}$ is the union of all the $10$ rooted trees.
\end{enumerate}
\begin{figure*}[!t]
			\centering	
			\begin{subfigure}{.40\textwidth}
				\centering
				\includegraphics[width=0.9\columnwidth]{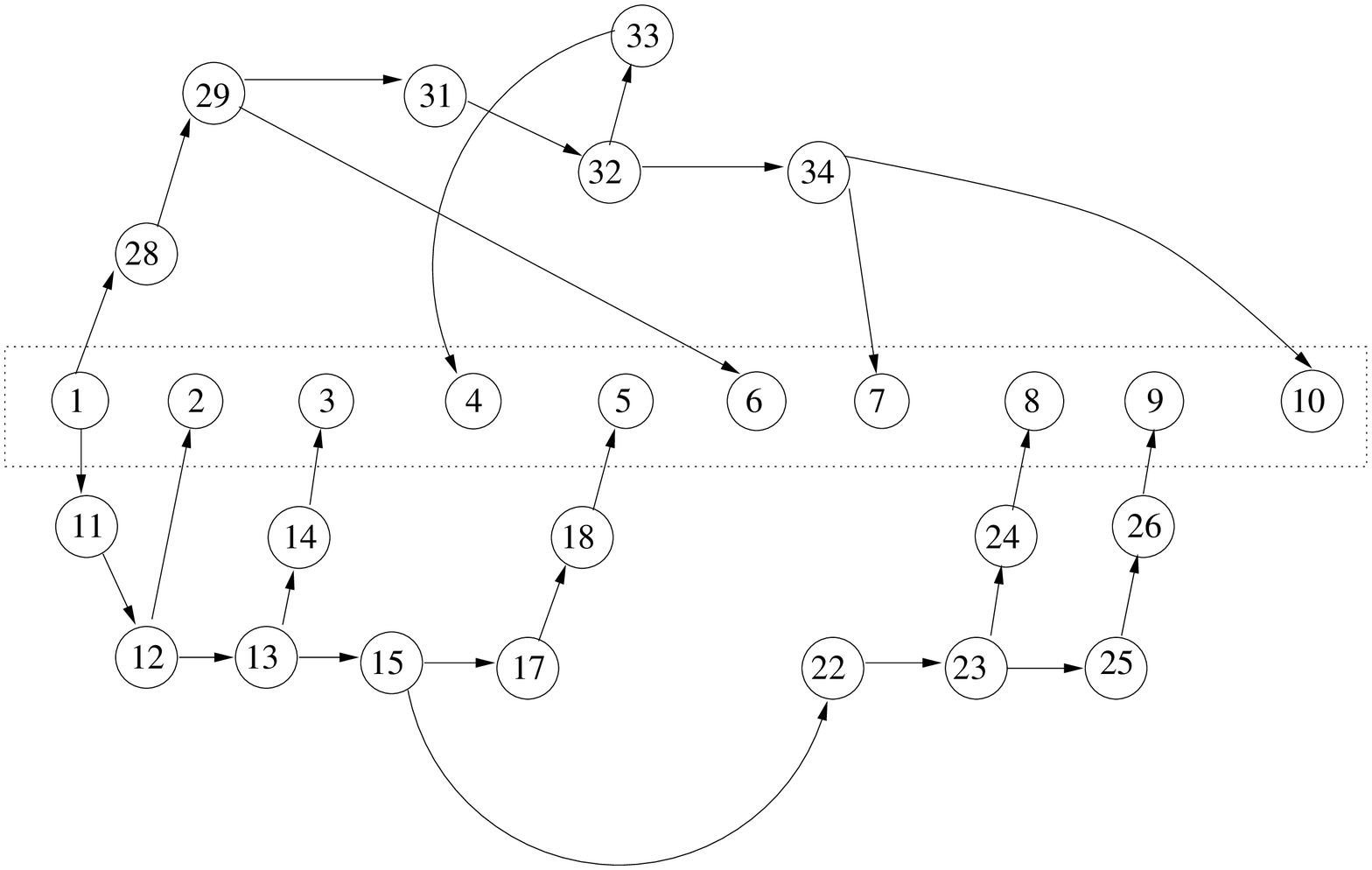}
				\caption{}
				\label{rt21}
			\end{subfigure}%
			\begin{subfigure}{.40\textwidth}
				\centering
				\includegraphics[width=0.9\columnwidth]{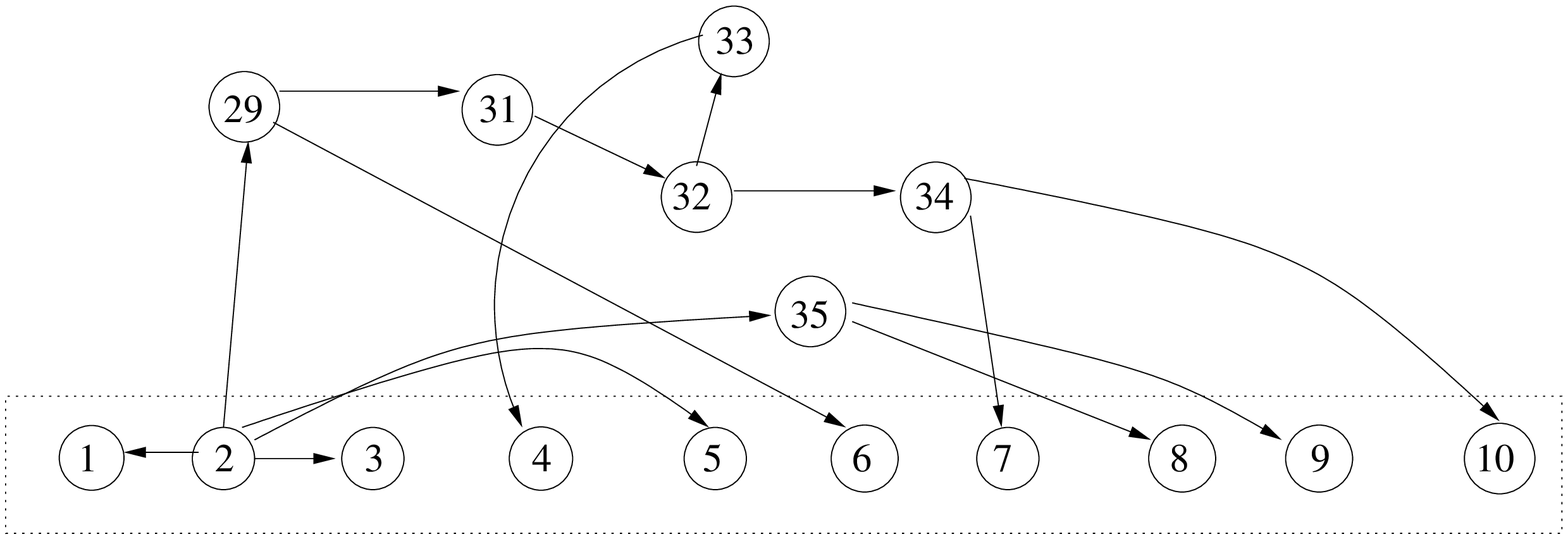}
				\caption{}
				\label{rt22}
			\end{subfigure}
			\begin{subfigure}{.40\textwidth}
				\centering
				\includegraphics[width=0.9\columnwidth]{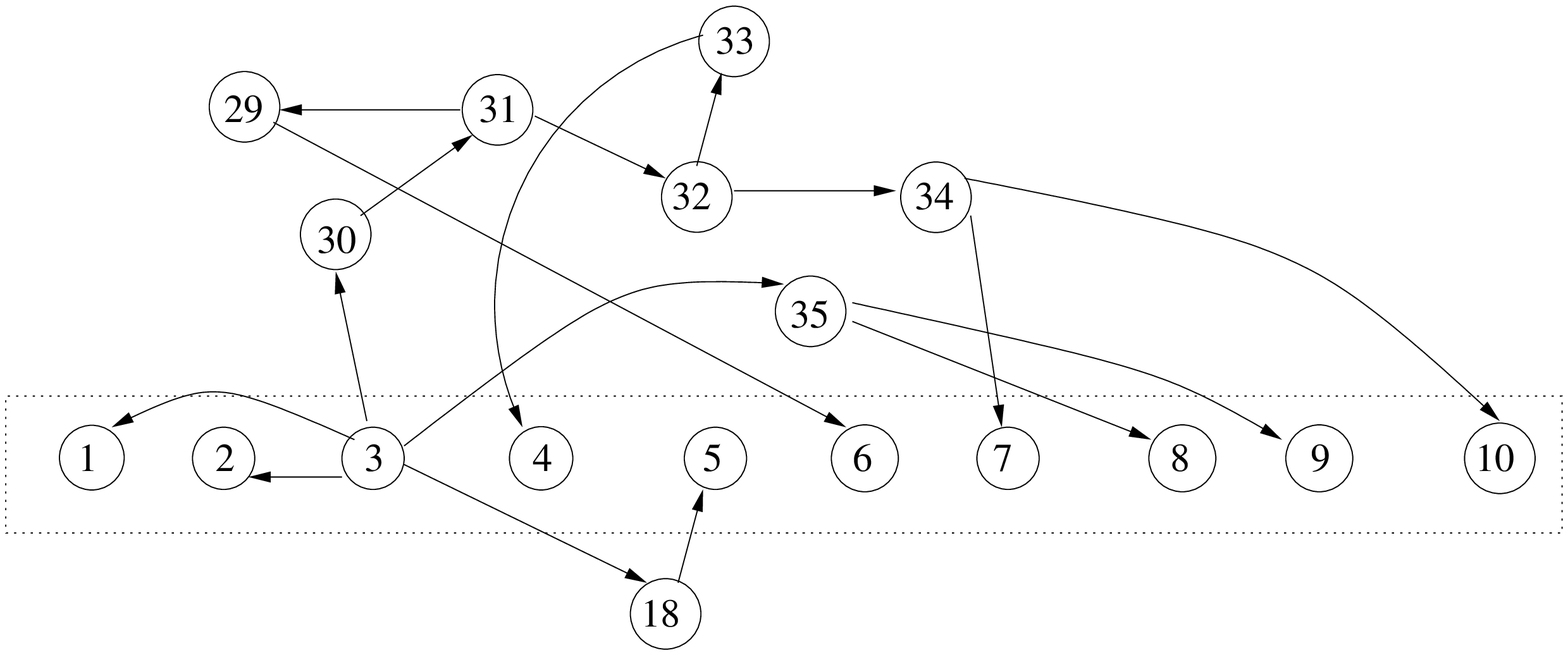}
				\caption{}
				\label{rt23}
			\end{subfigure}
			\begin{subfigure}{.40\textwidth}
				\centering
				\includegraphics[width=0.9\columnwidth]{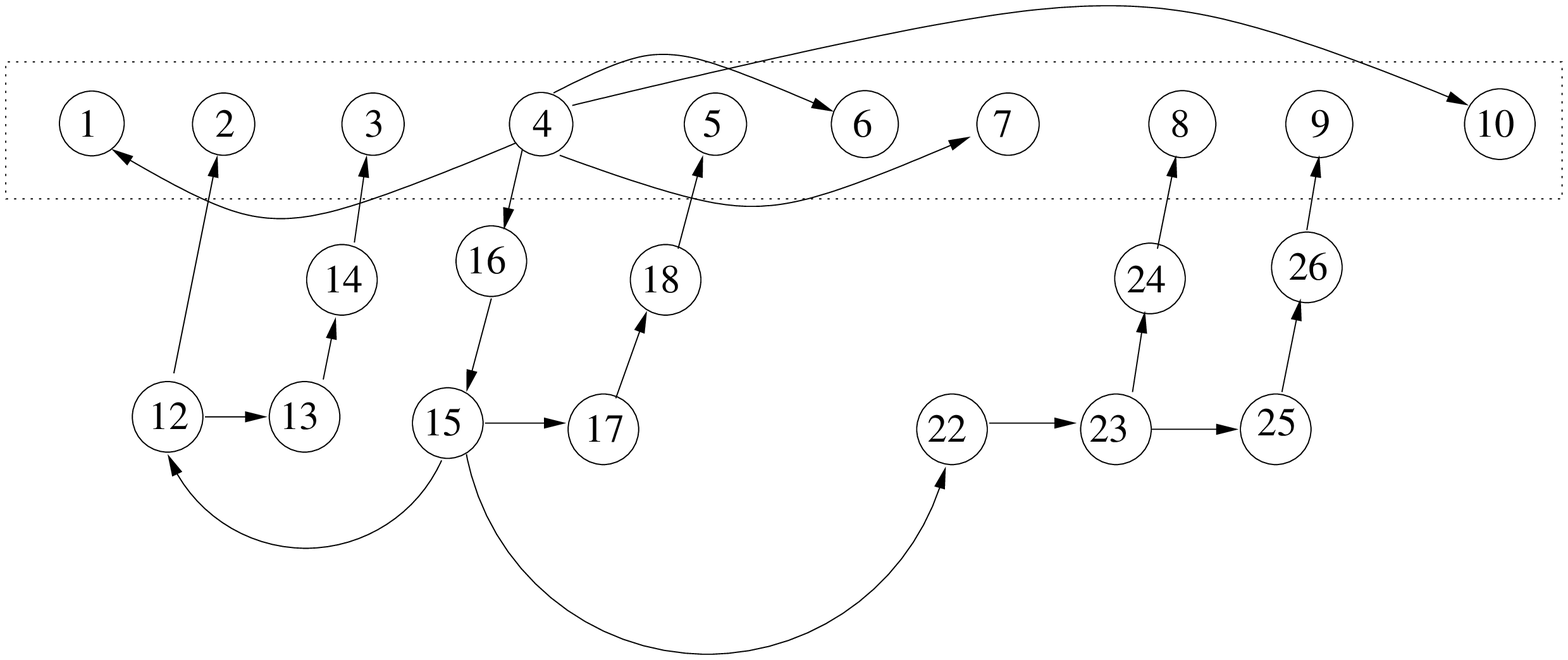}
				\caption{}
				\label{rt24}
			\end{subfigure}%
			\begin{subfigure}{.40\textwidth}
				\centering
				\includegraphics[width=0.9\columnwidth]{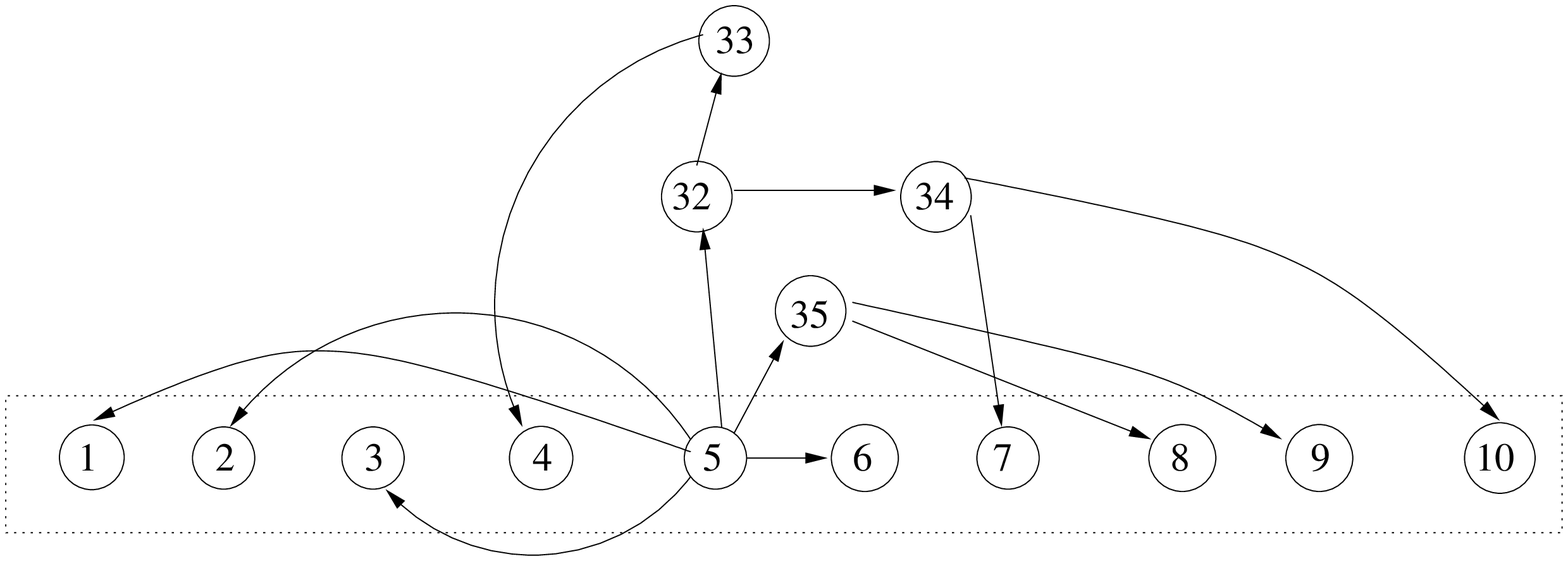}
				\caption{}
				\label{rt25}
			\end{subfigure}
			\begin{subfigure}{.40\textwidth}
				\centering
				\includegraphics[width=0.9\columnwidth]{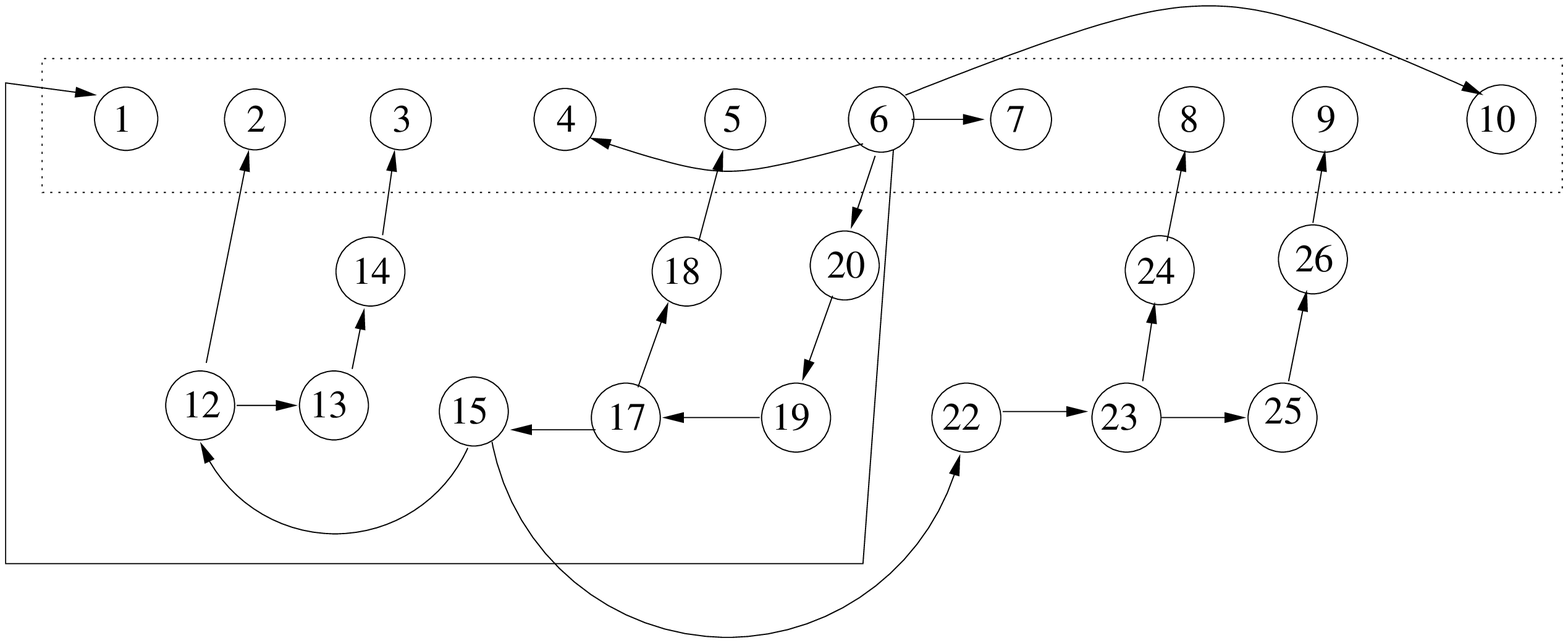}
				\caption{}
				\label{rt26}
			\end{subfigure}%
			\begin{subfigure}{.40\textwidth}
				\centering
				\includegraphics[width=0.9\columnwidth]{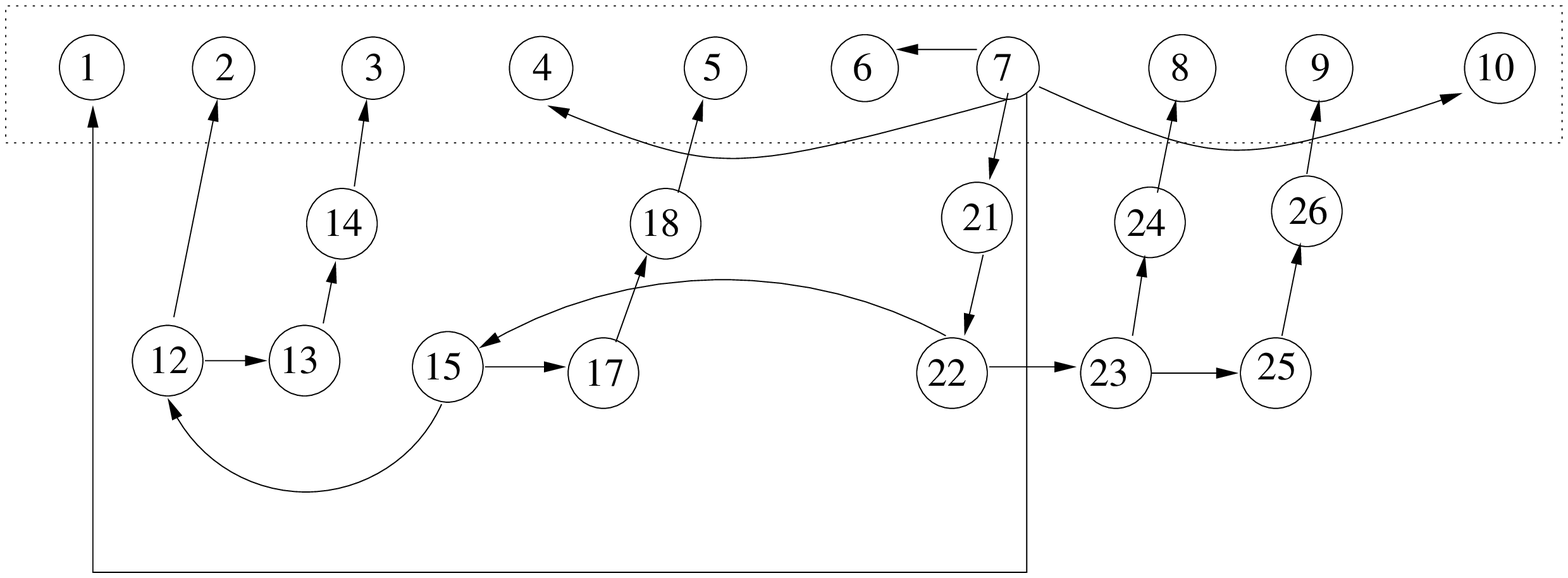}
				\caption{}
				\label{rt27}
			\end{subfigure}
			\begin{subfigure}{.40\textwidth}
				\centering
				\includegraphics[width=0.9\columnwidth]{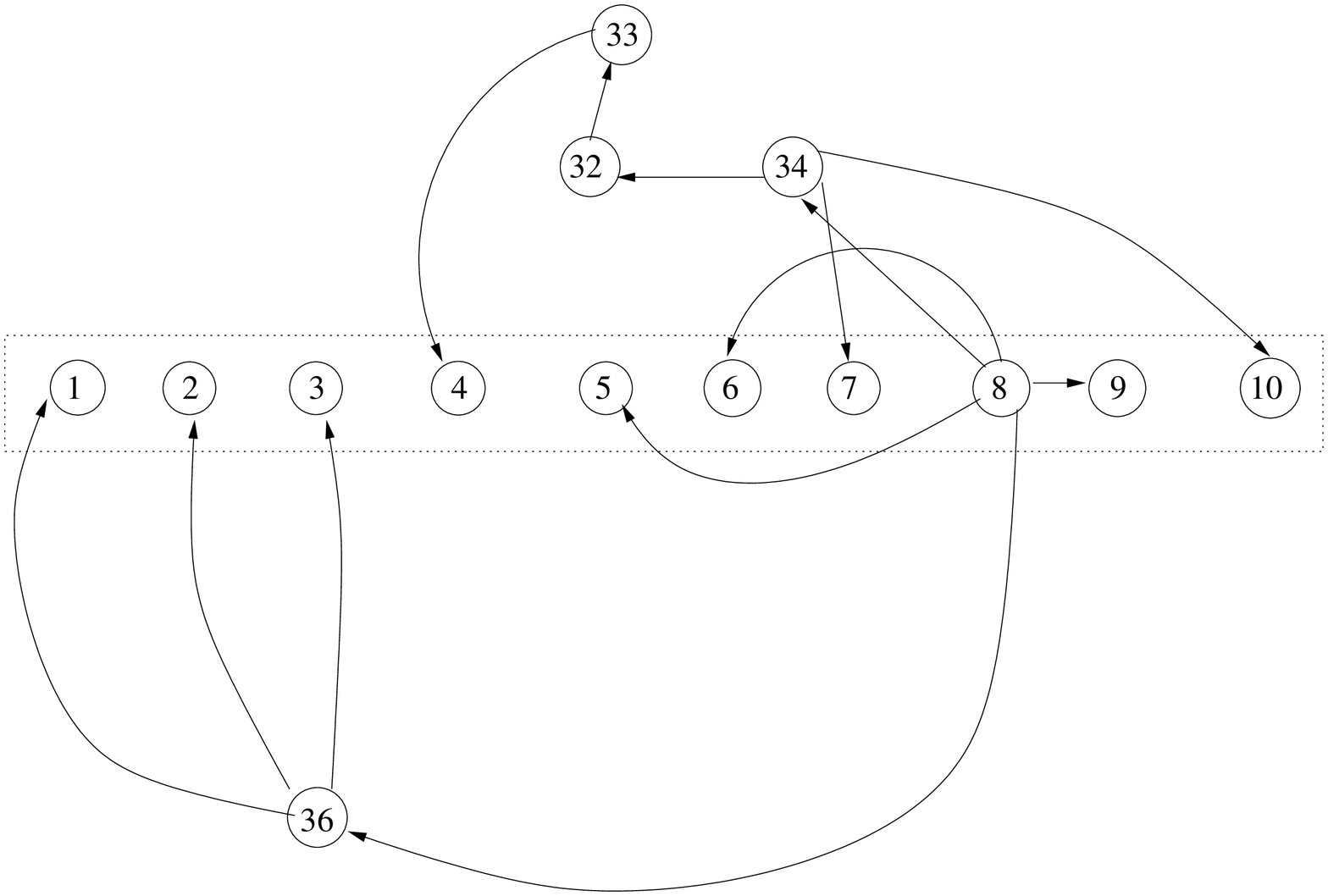}
				\caption{}
				\label{rt28}
			\end{subfigure}
			\begin{subfigure}{.40\textwidth}
				\centering
				\includegraphics[width=0.9\columnwidth]{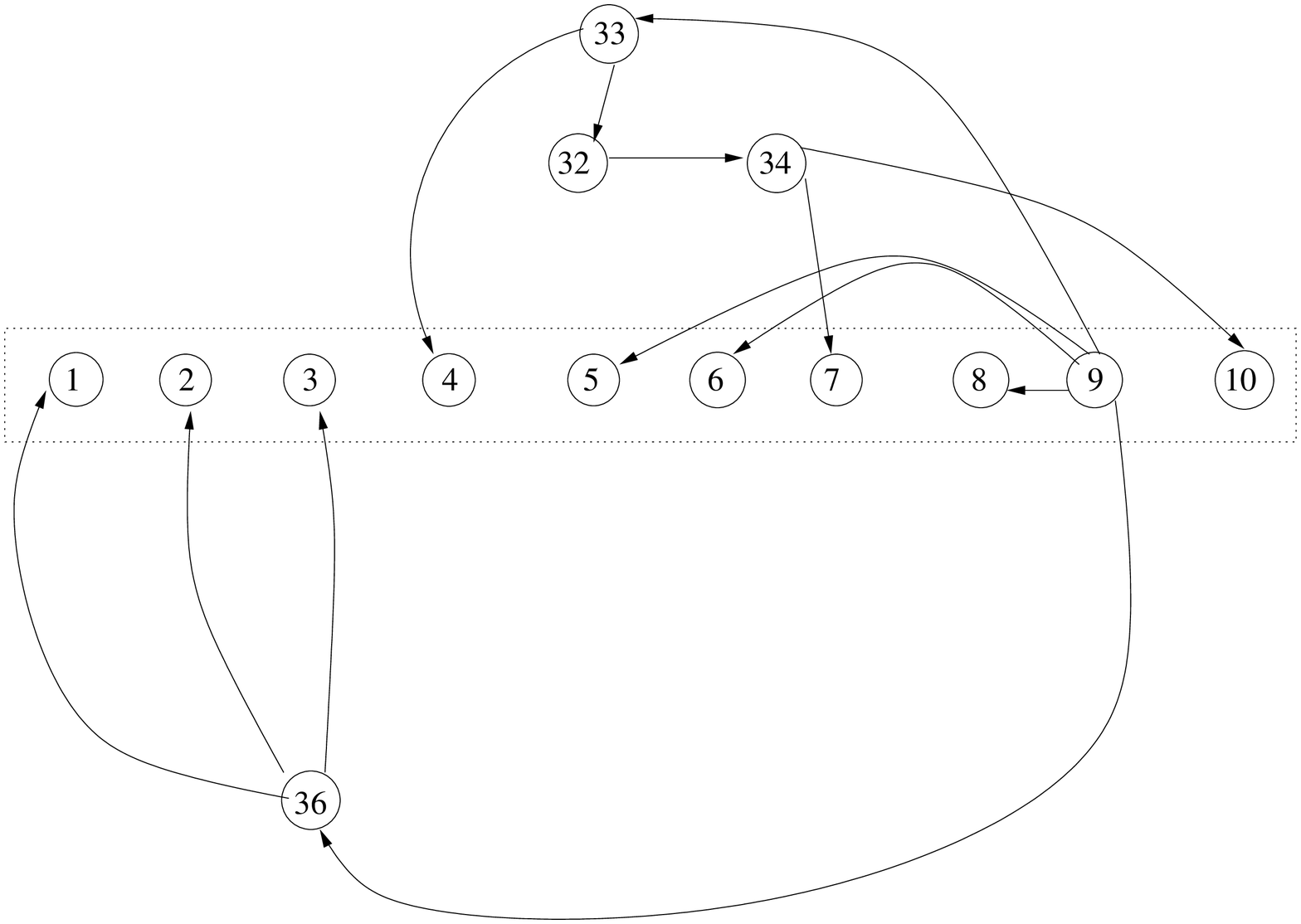}
				\caption{}
				\label{rt29}
			\end{subfigure}%
			\begin{subfigure}{.40\textwidth}
				\centering
				\includegraphics[width=0.9\columnwidth]{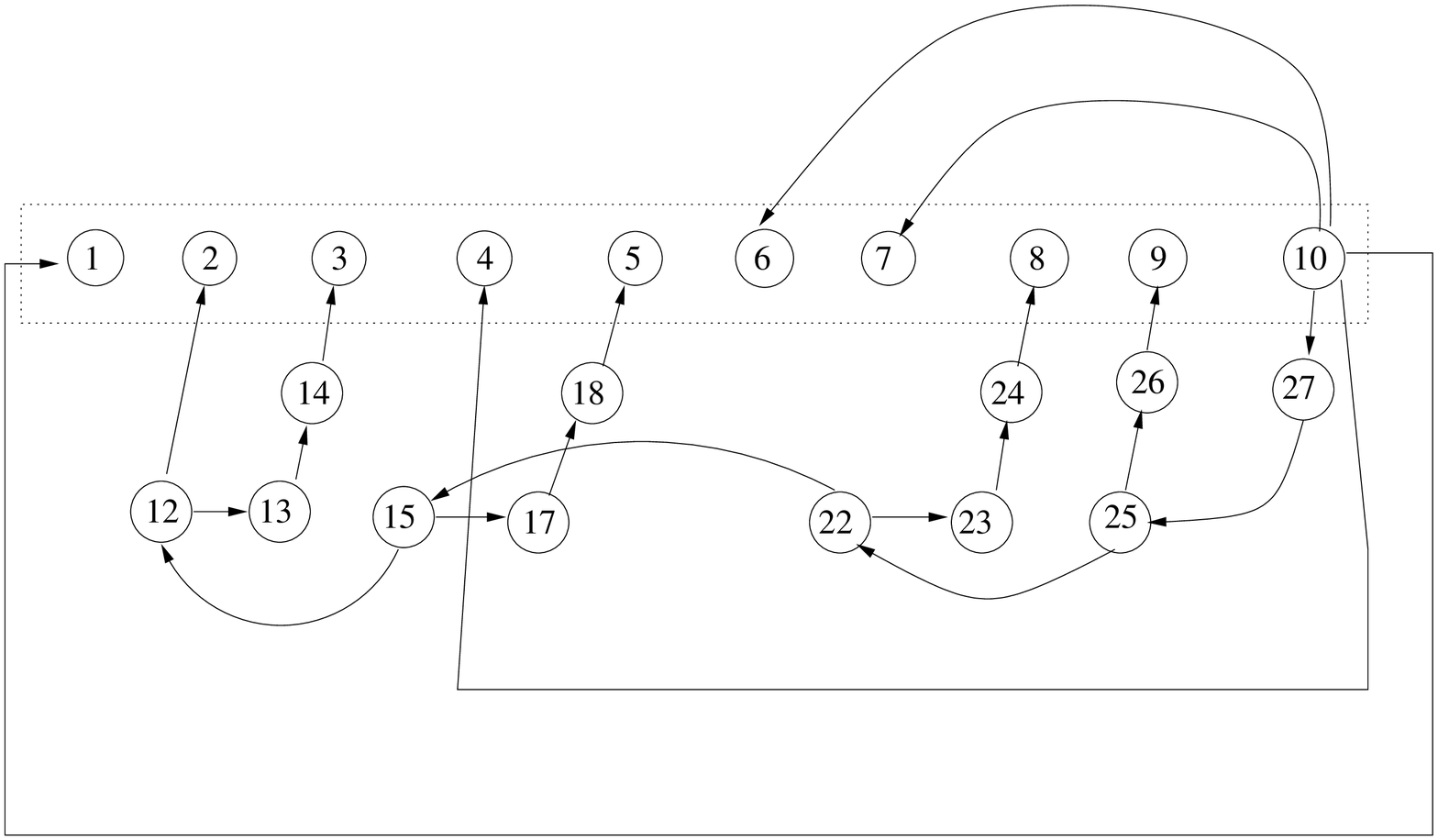}
				\caption{}
				\label{rt210}
			\end{subfigure}
			\caption{Figures showing rooted trees of inner vertices $ 1,2,3,4,5,6,7,8,9,10 $ of $ \mathcal{G} $, respectively.}
		\end{figure*}
		Now,\\ $ V_I = \lbrace 1,2,3,4,5,6,7,8,9,10 \rbrace $.\\
		$ V_{NI} = \lbrace11,12,13,14,\dots,34,35,36\rbrace $.\\
		$ V_{OC} = \lbrace 12,13,15,17,22,23,25,29,31,32,33,34\rbrace $.\\
		$ V_{MOCG}=\lbrace 12,13,15,17,22,23,25,32,33,34 \rbrace $.\\
		There is one isolated MOCG with CCV $ 32 $ formed by the union of two outer cycles with vertex sets $ \lbrace 32,33 \rbrace $ and $ \lbrace32,34\rbrace $.\\There are two non-isolated MOCGs one with CCV $ 15 $ which is formed by the union of three outer cycles with vertex sets $ \lbrace12,13,15 \rbrace $, $ \lbrace15,17\rbrace $ and $ \lbrace15,22\rbrace $ respectively and one with CCV $ 22 $ which is formed by union of two outer cycles with vertex sets $ \lbrace15,22\rbrace $ and $ \lbrace22,23,25\rbrace $. It is seen that these two MOCGs have the outer cycle with vertex set $ \lbrace15,22\rbrace $ in common. Now \textit{Construction} \ref{cons2} is carried out as follows.
		\begin{enumerate}
			\item $ W_I= x_1\oplus x_2\oplus x_3\oplus x_4\oplus x_5\oplus x_6\oplus x_7\oplus x_8\oplus x_9\oplus x_{10} $.
			\item The vertices in $ V_{NI}\backslash V_{OC} $ are $ 11 $, $ 14 $, $ 16 $, $ 18 $, $ 19 $, $ 20 $, $ 21 $, $ 24 $, $ 26 $, $ 27 $, $ 28 $, $ 30 $, $ 35 $ and $ 36 $.
			\begin{itemize}
				\item $ j=11 $. $ N^+_{\mathcal{G}}(11)\backslash V_{OC} = \phi $. Hence, $ W_{11}=x_{11} $. 
				\item $ j=14 $. $ N^+_{\mathcal{G}}(14)\backslash V_{OC} = \lbrace3 \rbrace $. Hence, $ W_{14}=x_{14} \oplus x_3 $.
				\item $ j=16 $. $ N^+_{\mathcal{G}}(16)\backslash V_{OC} = \phi $. Hence, $ W_{16}=x_{16} $.
				\item $ j=18 $. $ N^+_{\mathcal{G}}(18)\backslash V_{OC} = \lbrace5\rbrace $. Hence, $ W_{18}=x_{18}\oplus x_5 $.
				\item $ j=19 $. $ N^+_{\mathcal{G}}(19)\backslash V_{OC} = \phi $. Hence, $ W_{19}=x_{19} $.
				\item $ j=20 $. $ N^+_{\mathcal{G}}(20)\backslash V_{OC} = \lbrace 19 \rbrace $. Hence, $ W_{20}=x_{20}\oplus x_{19} $.
				\item $ j=21 $. $ N^+_{\mathcal{G}}(21)\backslash V_{OC} = \phi $. Hence, $ W_{21}=x_{21} $.
				\item $ j=24 $. $ N^+_{\mathcal{G}}(24)\backslash V_{OC} = \lbrace 8\rbrace $. Hence, $ W_{24}=x_{24}\oplus x_8 $.
				\item $ j=26 $. $ N^+_{\mathcal{G}}(26)\backslash V_{OC} = \lbrace9\rbrace $. Hence, $ W_{26}=x_{26}\oplus x_9 $.
				\item $ j=27 $. $ N^+_{\mathcal{G}}(27)\backslash V_{OC} = \phi $. Hence, $ W_{27}=x_{27} $.
				\item $ j=28 $. $ N^+_{\mathcal{G}}(28)\backslash V_{OC} = \phi $. Hence, $ W_{28}=x_{28} $.
				\item $ j=30 $. $ N^+_{\mathcal{G}}(30)\backslash V_{OC} = \phi $. Hence, $ W_{30}=x_{30} $.
				\item $ j=35 $. $ N^+_{\mathcal{G}}(35)\backslash V_{OC} = \lbrace8,9\rbrace $. Hence, $ W_{35}=x_{35}\oplus x_8\oplus x_9 $.
				\item $ j=36 $. $ N^+_{\mathcal{G}}(36)\backslash V_{OC} = \lbrace1,2,3\rbrace $. Hence, $ W_{36}=x_{36}\oplus x_1\oplus x_2\oplus x_3 $.
			\end{itemize}
			\item The isolated MOCG with CCV $ 32 $ is union of two outer cycles and the non-isolated MOCG with CCV $ 22 $ is also union of two outer cycles.
			\begin{itemize}
				\item $ j_1=32 $. $ j $ can be $ 33 $ or $ 34 $. Let $ j=33 $.\\
				$ N^+_{\mathcal{G}}(33) \backslash V_{OC}  = \lbrace4\rbrace$.\\
				$ N^+_{C}(33)=\lbrace32\rbrace $.\\Hence $  \lbrace N^+_{\mathcal{G}}(33) \backslash V_{OC}\rbrace \cup \lbrace N^+_{C}(33) \backslash  \lbrace 32 \rbrace \rbrace = \lbrace 4 \rbrace $. So, $ W_{33}=x_{33}\oplus x_4 $.
				\item $ j_1=22 $. $ j $ can be $ 15 $ or $ 25 $. Let $ j=25 $.\\
				$N^+_{\mathcal{G}}(25) \backslash V_{OC}  = \lbrace26\rbrace$.\\
				$ N^+_{C}(25)=\lbrace22\rbrace $.\\Hence $  \lbrace N^+_{\mathcal{G}}(25) \backslash V_{OC}\rbrace \cup \lbrace N^+_{C}(25) \backslash  \lbrace 22 \rbrace \rbrace = \lbrace 26 \rbrace $. So, $ W_{25}=x_{25}\oplus x_{26} $.
			\end{itemize} 
			\item The non-isolated MOCG with the CCV $ 15 $ is union of $ 3 $ (odd number) outer cycles.\\
			$ j_1=15 $.\\
			Pre-central cycle vertex set of $ 15 = \lbrace 13,17,22\rbrace $.\\
			Hence, let $ PV(15)=\lbrace 13,17\rbrace $.\\
			\begin{itemize}
				\item $ j= 17 $. $ N^+_{\mathcal{G}}(17) \backslash V_{OC}  = \lbrace18\rbrace$.\\
				$ N^+_{C}(17)=\lbrace15\rbrace $.\\Hence $  \lbrace N^+_{\mathcal{G}}(17) \backslash V_{OC}\rbrace \cup \lbrace N^+_{C}(17) \backslash  \lbrace 15 \rbrace \rbrace = \lbrace 18 \rbrace $. So, $ W_{17}=x_{17}\oplus x_{18} $.\\	
				\item $ j= 13 $. $ N^+_{\mathcal{G}}(13) \backslash V_{OC}  = \lbrace14\rbrace$.\\
				$ N^+_{C}(13)=\lbrace15\rbrace $.\\Hence $  \lbrace N^+_{\mathcal{G}}(13) \backslash V_{OC}\rbrace \cup \lbrace N^+_{C}(13) \backslash \lbrace 15 \rbrace \rbrace = \lbrace 14 \rbrace $. Therefore, $ W_{13}=x_{13}\oplus x_{14} $.
			\end{itemize}
			\item The remaining non-inner vertices are $ \lbrace 12,15,22,23,24,29,31,32,34 \rbrace $.
			\begin{itemize}
				\item $ j=12 $. $  N^+_{\mathcal{G}}(12) \backslash V_{OC}  = \lbrace2\rbrace $ and $ N^+_{C}(12) = \lbrace 13\rbrace $ and hence $ \lbrace N^+_{\mathcal{G}}(12) \backslash V_{OC} \rbrace \cup N^+_{C}(12)= \lbrace2,13\rbrace $. So, $ W_{12}=x_{12}\oplus x_2\oplus x_{13} $.
				\item $ j=15 $. $  N^+_{\mathcal{G}}(15) \backslash V_{OC}  = \phi $ and $ N^+_{C}(15) = \lbrace12,17,22\rbrace $ and hence $ \lbrace N^+_{\mathcal{G}}(15) \backslash V_{OC} \rbrace \cup N^+_{C}(15)= \lbrace12,17,22\rbrace $. So, $ W_{15}=x_{15}\oplus x_{12}\oplus x_{17}\oplus x_{22} $.
				\item $ j=22 $. $  N^+_{\mathcal{G}}(22) \backslash V_{OC}  = \phi $ and $ N^+_{C}(22) = \lbrace 15,23\rbrace $ and hence $ \lbrace N^+_{\mathcal{G}}(22) \backslash V_{OC} \rbrace \cup N^+_{C}(22)= \lbrace15,23\rbrace $. So, $ W_{22}=x_{22}\oplus x_{15}\oplus x_{23} $.
				\item $ j=23 $. $  N^+_{\mathcal{G}}(23) \backslash V_{OC}  = \lbrace24\rbrace $ and $ N^+_{C}(23) = \lbrace 25\rbrace $ and hence $ \lbrace N^+_{\mathcal{G}}(23) \backslash V_{OC} \rbrace \cup N^+_{C}(23)= \lbrace24,25\rbrace $. So, $ W_{23}=x_{23}\oplus x_{24}\oplus x_{25} $.
				\item $ j=29 $. $  N^+_{\mathcal{G}}(29) \backslash V_{OC}  = \lbrace6\rbrace $ and $ N^+_{C}(29) = \lbrace 31\rbrace $ and hence $ \lbrace N^+_{\mathcal{G}}(29) \backslash V_{OC} \rbrace \cup N^+_{C}(29)= \lbrace6,31\rbrace $. So, $ W_{29}=x_{29}\oplus x_6\oplus x_{31} $.
				\item $ j=31 $. $  N^+_{\mathcal{G}}(31) \backslash V_{OC}  =  \phi $ and $ N^+_{C}(31) = \lbrace 29\rbrace $ and hence $ \lbrace N^+_{\mathcal{G}}(31) \backslash V_{OC} \rbrace \cup N^+_{C}(31)= \lbrace29\rbrace $. So, $ W_{31}=x_{31}\oplus x_{29} $.
				\item $ j=32 $. $  N^+_{\mathcal{G}}(32) \backslash V_{OC}  = \phi $ and $ N^+_{C}(32) = \lbrace 33,34\rbrace $ and hence $ \lbrace N^+_{\mathcal{G}}(32) \backslash V_{OC} \rbrace \cup N^+_{C}(32)= \lbrace33,34\rbrace $. So, $ W_{32}=x_{32}\oplus x_{33}\oplus x_{34} $.
				\item $ j=34 $. $  N^+_{\mathcal{G}}(34) \backslash V_{OC}  = \lbrace7,10\rbrace $ and $ N^+_{C}(34) = \lbrace 32\rbrace $ and hence $ \lbrace N^+_{\mathcal{G}}(34) \backslash V_{OC} \rbrace \cup N^+_{C}(34)= \lbrace7,10,32\rbrace $. So, $ W_{34}=x_{34}\oplus x_7\oplus x_{10}\oplus x_{32} $.
			\end{itemize}
		\end{enumerate}
		The obtained index code symbols are listed as follows.\\
			 $ W_I= x_1\oplus x_2\oplus x_3\oplus x_4\oplus x_5\oplus x_6\oplus  x_7\oplus x_8\oplus x_9\oplus x_{10} $.\\
			$ W_{11}=x_{11} $.\\
$			 W_{12}=x_{12}\oplus x_2\oplus x_{13} $.\\
$			 W_{13}=x_{13}\oplus x_{14} $.\\
$			 W_{14}=x_{14} \oplus x_3 $.\\
$			 W_{15}=x_{15}\oplus x_{12}\oplus x_{17}\oplus x_{22} $.\\
$			 W_{16}=x_{16} $.\\
$			 W_{17}=x_{17}\oplus x_{18} $.\\
			$ W_{18}=x_{18}\oplus x_5 $.\\
$			 W_{19}=x_{19} $.\\
$			 W_{20}=x_{20}\oplus x_{19} $.\\
$			 W_{21}=x_{21} $.\\
$			 W_{22}=x_{22}\oplus x_{15}\oplus x_{23} $.\\
$			 W_{23}=x_{23}\oplus x_{24}\oplus x_{25} $.\\
$			 W_{24}=x_{24}\oplus x_8 $.\\
$			 W_{25}=x_{25}\oplus x_{26} $.\\
$			 W_{26}=x_{26}\oplus x_9 $.\\
$			 W_{27}=x_{27} $.\\
$			 W_{28}=x_{28} $.\\
$			 W_{29}=x_{29}\oplus x_6\oplus x_{31} $.\\
$			 W_{30}=x_{30} $.\\
$			 W_{31}=x_{31}\oplus x_{29} $.\\
			 $ W_{32}=x_{32}\oplus x_{33}\oplus x_{34} $.\\
			 $ W_{33}=x_{33}\oplus x_4 $.\\
			 $ W_{34}=x_{34}\oplus x_7\oplus x_{10}\oplus x_{32} $.\\
			 $ W_{35}=x_{35}\oplus x_8\oplus x_9 $.\\
			 $ W_{36}=x_{36}\oplus x_1\oplus x_2\oplus x_3 $.
		The decoding of messages corresponding to non-inner vertices can be done using index code symbols corresponding to the respective non-inner vertices. The decoding of messages corresponding to inner vertices using \textit{Algorithm} \ref{algo2} is as follows.
		\begin{itemize}
			\item $ i=1 $.\\
			$ V_{NIC}(1)=\lbrace 12,13,15,17,22,23,25,29,31,32,33,34\rbrace $.\\
			$ V'_{NI}(1)= V_{NIC}(1) $ and $ V_{NI}(1)\backslash V_{NIC}(1) =\lbrace 11,14,18,24,26,28\rbrace $.\\
			$ Z_1 = W_I\oplus W_{11}\oplus W_{12}\oplus W_{13}\oplus W_{14}\oplus W_{15}\oplus W_{17}\oplus W_{18}\oplus W_{22}\oplus W_{23}\oplus W_{24} \oplus W_{25}\oplus W_{26}\oplus W_{28}\oplus W_{29}\oplus W_{31}\oplus W_{32}\oplus W_{33}\oplus W_{34}$ which results in $ Z_1=x_1\oplus x_{11}\oplus x_{28} $ and user $ 1 $ has $ x_{11} $ and $ x_{28} $ in its side information.
			\item $ i=2 $.\\
			$ V_{NIC}(2)=\lbrace 29,31,32,33,34\rbrace $.\\
			$ V'_{NI}(2)= V_{NIC}(2) $ and $ V_{NI}(2)\backslash V_{NIC}(2) =\lbrace 35\rbrace $.\\
			$ Z_2 = W_I\oplus W_{29}\oplus W_{31}\oplus W_{32}\oplus W_{33}\oplus W_{34} \oplus  W_{28}$ which results in $ Z_2=x_2\oplus x_{1}\oplus x_{3}\oplus x_5 $ and user $ 2 $ has $ x_{1} $, $ x_3 $ and $ x_{5} $ in its side information.
			\item $ i=3 $.\\
			$ V_{NIC}(3)=\lbrace 29,31,32,33,34\rbrace $.\\
			$ V'_{NI}(3)= V_{NIC}(3) $ and $ V_{NI}(3)\backslash V_{NIC}(3) =\lbrace18,30,35\rbrace $.\\
			$ Z_3 = W_I\oplus W_{18}\oplus W_{29}\oplus W_{30}\oplus W_{31}\oplus W_{32}\oplus W_{33}\oplus W_{34} \oplus  W_{35}$ which results in $ Z_3=x_3\oplus x_{1}\oplus x_{2}\oplus x_{18}\oplus x_{30}\oplus x_{35} $ and user $ 3 $ has $ x_{1} $, $ x_2 $, $ x_{18}$, $ x_{30} $ and $ x_{35} $ in its side information.
			\item $ i=4 $.\\
			$ V_{NIC}(4)=\lbrace 12,13,15,17,22,23,25\rbrace $.\\
			$ V'_{NI}(4)= V_{NIC}(4) $ and $ V_{NI}(4)\backslash V_{NIC}(4) =\lbrace14,16,18,24,26\rbrace $.\\
			$ Z_4 = W_I\oplus W_{12}\oplus W_{13}\oplus W_{14}\oplus W_{15}\oplus W_{16}\oplus W_{17}\oplus W_{18} \oplus W_{22}\oplus W_{23}\oplus W_{24}\oplus W_{25}\oplus W_{26}$ which results in $ Z_4=x_4\oplus x_{1}\oplus x_{6}\oplus x_{7}\oplus x_{16}$ and user $ 4 $ has $ x_{1} $, $ x_6 $, $ x_{7} $ and $ x_{16} $ in its side information.
			\item $ i=5 $.\\
			$ V_{NIC}(5)=\lbrace 32,33,34\rbrace $.\\
			$ V'_{NI}(5)= V_{NIC}(5) $ and $ V_{NI}(5)\backslash V_{NIC}(5) =\lbrace35\rbrace $.\\
			$ Z_5 = W_I\oplus W_{32}\oplus W_{33}\oplus W_{34}\oplus W_{35} $ which results in $ Z_5=x_5\oplus x_{1}\oplus x_{2}\oplus x_{3}\oplus x_6\oplus x_{35}$ and user $ 5 $ has $ x_{1} $, $ x_2 $, $ x_{3} $, $ x_6 $ and $ x_{35} $ in its side information.
			\item $ i=6 $.\\
			$ V_{NIC}(6)=\lbrace 12,13,15,17,22,23,25\rbrace $.\\
			$ V'_{NI}(6)= V_{NIC}(6) $ and $ V_{NI}(6)\backslash V_{NIC}(6) =\lbrace14,18,19,20,24,26\rbrace $.\\
			$ Z_6 = W_I\oplus W_{12}\oplus W_{13}\oplus W_{14}\oplus W_{15}\oplus W_{17}\oplus W_{18}\oplus W_{19}\oplus W_{20}\oplus W_{22}\oplus W_{23}\oplus W_{24}\oplus W_{25}\oplus W_{26} $ which results in $ Z_6=x_6\oplus x_{1}\oplus x_{4}\oplus x_{7}\oplus x_{10}\oplus x_{20}$ and user $ 6 $ has $ x_{1} $, $ x_4 $, $ x_7 $, $ x_{10} $ and $ x_{20} $ in its side information.
			\item $ i=7 $.\\
			$ V_{NIC}(7)=\lbrace 12,13,15,17,22,23,25\rbrace $.\\
			$ V'_{NI}(7)= V_{NIC}(7) $ and $ V_{NI}(7)\backslash V_{NIC}(7) =\lbrace14,18,21,24,26\rbrace $.\\
			$ Z_7 = W_I\oplus W_{12}\oplus W_{13}\oplus W_{14}\oplus W_{15}\oplus W_{17}\oplus W_{18}\oplus W_{21}\oplus W_{22}\oplus W_{23}\oplus W_{24}\oplus W_{25}\oplus W_{26} $ which results in $ Z_7=x_7\oplus x_{1}\oplus x_{4}\oplus x_{6}\oplus x_{10}\oplus x_{21}$ and user $ 7 $ has $ x_{1} $, $ x_4 $, $ x_{6} $, $ x_{10} $ and $ x_{21} $ in its side information.
			\item $ i=8 $.\\
			$ V_{NIC}(8)=\lbrace 32,33,34\rbrace $.\\
			$ V'_{NI}(8)= V_{NIC}(8) $ and $ V_{NI}(8)\backslash V_{NIC}(8) =\lbrace36\rbrace $.\\
			$ Z_8 = W_I\oplus W_{32}\oplus W_{33}\oplus W_{34}\oplus W_{36} $ which results in $ Z_8=x_8\oplus x_{5}\oplus x_{6}\oplus x_{9}\oplus x_{36}$ and user $ 8 $ has $ x_{5} $, $ x_6 $, $ x_{9} $ and $ x_{36} $ in its side information.
			\item $ i=9 $.\\
			$ V_{NIC}(9)=\lbrace 32,33,34\rbrace $.\\
			$ V'_{NI}(9)= V_{NIC}(9) $ and $ V_{NI}(9)\backslash V_{NIC}(9) =\lbrace36\rbrace $.\\
			$ Z_9 = W_I\oplus W_{32}\oplus W_{33}\oplus W_{34}\oplus W_{36} $ which results in $ Z_9=x_9\oplus x_{5}\oplus x_{6}\oplus x_{8}\oplus x_{36}$ and user $ 9 $ has $ x_{5} $, $ x_6 $, $ x_{8} $ and $ x_{36} $ in its side information.
			\item $ i=10 $.\\
			$ V_{NIC}(10)=\lbrace 12,13,15,17,22,23,25\rbrace $.\\
			$ V'_{NI}(10)= V_{NIC}(10) $ and $ V_{NI}(10)\backslash V_{NIC}(10) =\lbrace14,18,24,26,27\rbrace $.\\
			$ Z_{10} =  W_I\oplus W_{12}\oplus W_{13}\oplus W_{14}\oplus W_{15}\oplus W_{17}\oplus W_{18}\oplus W_{22}\oplus W_{23}\oplus W_{24}\oplus W_{25}\oplus W_{26}\oplus W_{27} $ which results in $ Z_{10}=x_{10}\oplus x_{1}\oplus x_{6}\oplus x_{7}\oplus x_{24}$ and user $ 10 $ has $ x_{1} $, $ x_6 $, $ x_{7} $ and $ x_{24} $ in its side information.
		\end{itemize}
		\end{ex}Thus an index code obtained by using \textit{Construction} \ref{cons2} on an IC structure with outer cycles is decodable using \textit{Algorithm} \ref{algo2}.
\section{Conclusion}
\label{sec_con}
	In this paper, index code construction and a decoding algorithm are given for any IC structure with outer cycles. 
In \cite{vaddi}, for a particular type of structure on the outer cycles for a IC structure optimal length index code is constructed. Investigating the optimality of the index codes constructed at least for some subclasses  is an interesting open problem. 

\end{document}